\documentclass[10pt,twocolumn,letterpaper]{article}

\usepackage{cvpr}              
\definecolor{cvprblue}{rgb}{0.21,0.49,0.74}
\usepackage[accsupp]{axessibility}  
\usepackage[pagebackref,breaklinks,colorlinks,allcolors=cvprblue]{hyperref}
\usepackage{xcolor}
\usepackage{multirow}
\usepackage{amsthm}
\usepackage[table]{xcolor}
\usepackage[inkscapelatex=false]{svg}

\newtheorem{theorem}{Theorem}

\newcommand{\name}{\textrm{mmDiff}\xspace}

\def\paperID{} 
\def\confName{CVPR}
\def\confYear{2026}

\title{\name: A Noise-Robust Differentiable Ray-Tracing Framework \\for mmWave Scene Calibration and Channel Prediction}

\author{Haofan Lu\\
UCLA\\
{\tt\small haofan@cs.ucla.edu}
\and
Yadi Cao\\
UC San Diego\\
{\tt\small yac066@ucsd.edu}
\and
Wanghao Yi\\
UCLA\\
{\tt\small wanghao24@ucla.edu}
\and
Omid Abari\\
UCLA\\
{\tt\small omid@cs.ucla.edu}
}

\begin{document}
\maketitle
\begin{abstract}
3D reconstruction techniques such as LiDAR scanning and photogrammetry have made it practical to build detailed geometric models of real-world environments. Such reconstructed models can potentially serve as the foundation for wireless digital twins and support network planning and optimization. The core challenge is that reconstructed models inevitably contain geometric artifacts such as holes and noisy surfaces, and wireless simulation is highly sensitive to such noise. To solve this problem, we propose a differentiable directional scattering model to approximate the noise-sensitive specular reflection. This approximation smoothly distributes reflected power among nearby ray directions, making the simulator inherently robust to local geometric artifacts in the reconstructed model. We prove mathematically that this approximation preserves asymptotic path-gain accuracy. Building on this idea, we propose mmDiff, an end-to-end differentiable framework for calibrating material properties from sparse mmWave measurements and predicting mmWave channels. We evaluate mmDiff on both real-world and synthetic datasets, and demonstrate its superior performance over prior methods using pure specular reflection in noisy reconstructed geometry. We open-source our implementation at \url{https://github.com/LuHaofan/mmDiff}.
\end{abstract}

\vspace{-5pt}
\section{Introduction}
\label{sec:intro}
The rapid proliferation of consumer electronics, including smartphones, mixed reality headsets, and autonomous vehicles, has driven increasing demands for high-throughput, low-latency wireless connectivity. Millimeter wave (mmWave) technology~\cite{rappaportOverviewMillimeterWave2017}, operating across the 30–300 GHz frequency bands, offers extensive bandwidth that enables multi-gigabit data rates and sub-millisecond latency.

However, mmWave signals attenuate rapidly with distance and are easily blocked by everyday objects, limiting reliable communication range. To compensate, mmWave systems concentrate energy into narrow directional beams. This directionality demands precise alignment between the transmitter and receiver. In practice, the optimal beam directions are highly dependent on the physical environment: the layout of the space and the material composition. Conventional beam alignment algorithms assume no prior knowledge of the environment and must scan every possible direction, incurring significant overhead. Although many fast beam alignment methods have been proposed~\cite{agilelink,li2012efficient,tsang2011coding,yuan2015efficient,spacebeam}, they rarely exploit environmental priors to accelerate beam training.

Recent advances in 3D computer vision and neural representation have inspired researchers to model wireless radiation fields in the environment and use them for predicting wireless channels. There are two main lines of research:

Data-driven methods, such as \cite{NeRF2,GSRF,WRF-GS,WRF-GS+,RadSplatter,shenNeRFAPTNewNeRF2025}, adapt NeRFs and 3D Gaussian Splatting to the radio frequency domain, implicitly encoding the radiation field in neural network weights or Gaussian parameters. Without prior knowledge of the scene geometry, these methods require dense RF measurements that are hard to obtain in practice. Recent works~\cite{chen2024rfcanvas, zhangRF3DGSWirelessChannel2025} incorporate visual priors from cameras and LiDARs to relax this requirement, but the resulting representations remain implicit, coupling antenna pattern, polarization, path loss, and material properties in ways that limit explainability and adaptability to new transmitter/receiver configurations. 

Simulation-driven methods~\cite{wi-nert, diff-rt-calibration, jiaNeuralReflectanceFields2025, RadioTwin}, by contrast, decompose the scene into interpretable components: geometry, material properties, and transmitter/receiver configurations, and employ ray tracing to simulate radio propagation. This decomposition enables network operators to test diverse hypothetical configurations, making it well-suited for deployment planning. We follow the simulation-driven approach.

To bridge the gap between simulation and reality, simulation-driven methods require calibrating scene parameters, particularly the radio material properties. Early works manually assign material parameters and tune them using closed-form or simulated annealing optimization~\cite{kanhereCalibrationNYURay3D2023, charbonnierCalibrationRayTracingDiffuse2020a, priebeCalibratedBroadbandRay2012a}, which does not scale to large, complex scenes. The emerging paradigm is to make the ray tracing pipeline differentiable to enable gradient-based optimization~\cite{sionna}, which is more efficient and scalable with the acceleration of GPUs. This approach has been applied and verified at the sub-6GHz band~\cite{diff-rt-calibration}. However, mmWave systems differ significantly from sub-6GHz, where large antenna arrays are employed and signal scattering behavior is vastly different. Existing differentiable calibration frameworks are not designed for these characteristics, leaving differentiable calibration for mmWave an open problem.

In pursuing this direction, we encountered a key practical challenge: geometry noise in 3D reconstructed models. Prior differentiable calibration works assume a high-quality CAD model that faithfully replicates the real-world scene, but such models are rarely available in practice. While advances in 3D reconstruction technology have made scene modeling more accessible, the resulting models often exhibit flaws such as holes and noisy surfaces. These flaws cause critical propagation paths to be dropped or spurious paths to appear, degrading calibration performance. Manually fixing such flaws requires professional artists and days of effort, which is not scalable. We therefore ask: \textit{Can we design the simulator itself to be robust to geometry noise?} Our key insight is that this is particularly tractable at mmWave frequencies. At sub-6GHz, specular reflection dominates, and geometry noise is difficult to handle -- a small surface normal error directly translates to a large error in the reflected ray direction. At mmWave, however, scattering is stronger, and this opens an opportunity: we approximate the noise-sensitive specular reflection with a differentiable directional scattering model that smoothly distributes power among nearby ray directions, naturally absorbing the effect of geometry noise.

In this work, we present \name, the first end-to-end differentiable ray-tracing framework for mmWave channel calibration, enabling gradient-based optimization of radio material properties at scale. Our main contributions are: 

\begin{itemize}
\item We propose an end-to-end optimization framework for mmWave calibration using sparse mmWave channel measurements, bridging the gap of differentiable mmWave calibration. 
\item We introduce a differentiable approximation of the specular reflection. We prove that the approximation preserves the asymptotic accuracy of path-gain estimation.
\item We evaluate \name on real-world datasets, and observe 10.5 dB improvement in average received power prediction accuracy over the prior differentiable solver using specular reflection.
\end{itemize}

\vspace{-5pt}
\section{Preliminaries}
\label{sec:preliminaries}
\subsection{Angular Power Spectra}

\begin{figure}[t]
    \centering
    \includegraphics[width=0.7\columnwidth]{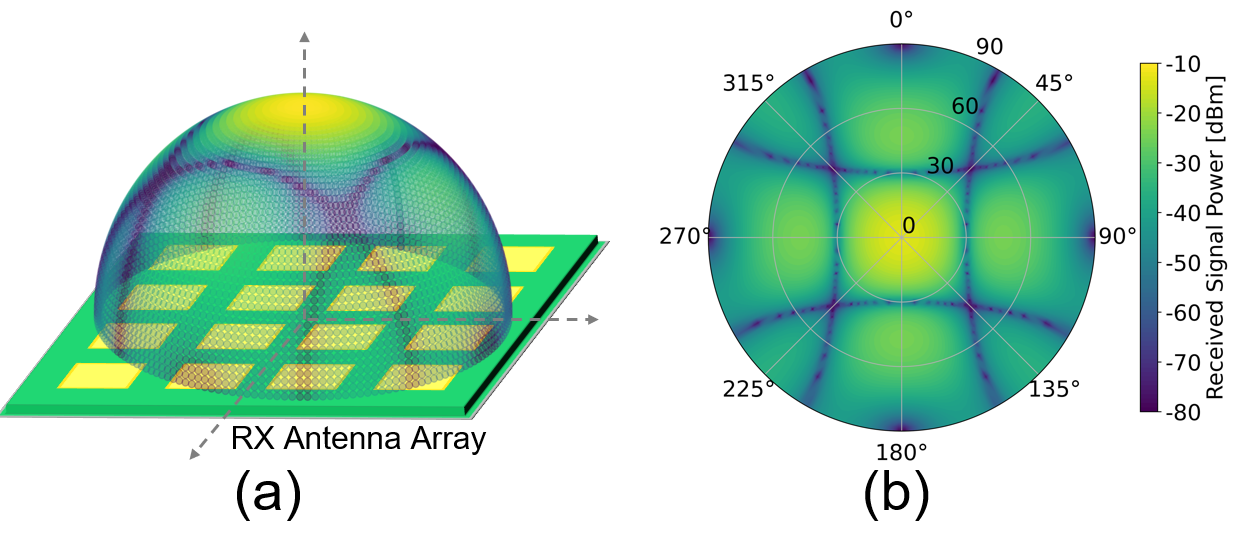}
    \caption{AoA power spectra. a) The 3D view of an AoA power spectra, overlying a 4x4 antenna array. b) The 2D polar plot of the AoA power spectra. The color indicates the received signal power from each direction.}
    \label{fig:aoa-spectrum}
    \vspace{-10pt}
\end{figure}
mmWave radios create highly directional beams to concentrate energy toward intended directions and overcome high path loss. Understanding how signal power arrives from different directions is therefore essential for characterizing channel behavior, optimizing coverage, and designing robust beam alignment strategies.

The angle-of-arrival (AoA) power spectrum provides this directional view by describing how much received signal power arrives from each direction in space, revealing the spatial structure of the wireless channel. Crucially, the AoA power spectrum is practically measurable: it can be obtained with a phased array or a mechanically rotated horn antenna that scans across directions, making it well-suited for mmWave scene calibration. 

\subsection{Physically-based rendering (PBR)}
Physically-based rendering (PBR) is a computer graphics technique that attempts to simulate how light behaves in the real world for rendering a view. It simulates light rays propagating and bouncing in the scene to produce realistic images. Recent efforts make the whole ray-tracing pipeline end-to-end differentiable, enabling the applications of optimizing scene parameters from images, known as inverse rendering~\cite{jakob2022mitsuba3}. This capability has been widely used in various tasks, including 3D reconstruction~\cite{Nicolet2021Large}, material or illumination estimation~\cite{chen2021dib}. Recently, this concept has been extended to the wireless domain, represented by NVIDIA's Sionna \cite{sionna}, a differentiable ray-tracing engine for wireless communication. Sionna builds the wireless PHY and MAC layers on top of the differentiable rendering engines~\cite{Jakob2022DrJit}, providing GPU-accelerated simulation for wireless coverage estimation. The latest version of Sionna implements two solvers for the radio map and channel impulse response (CIR), respectively. In this work, we propose the design of a new physics-grounded solver for the mmWave channel simulation. Enabling end-to-end calibration from the angular spectrum.

\begin{figure*}[t]
    \centering
    \includegraphics[width=0.9\textwidth]{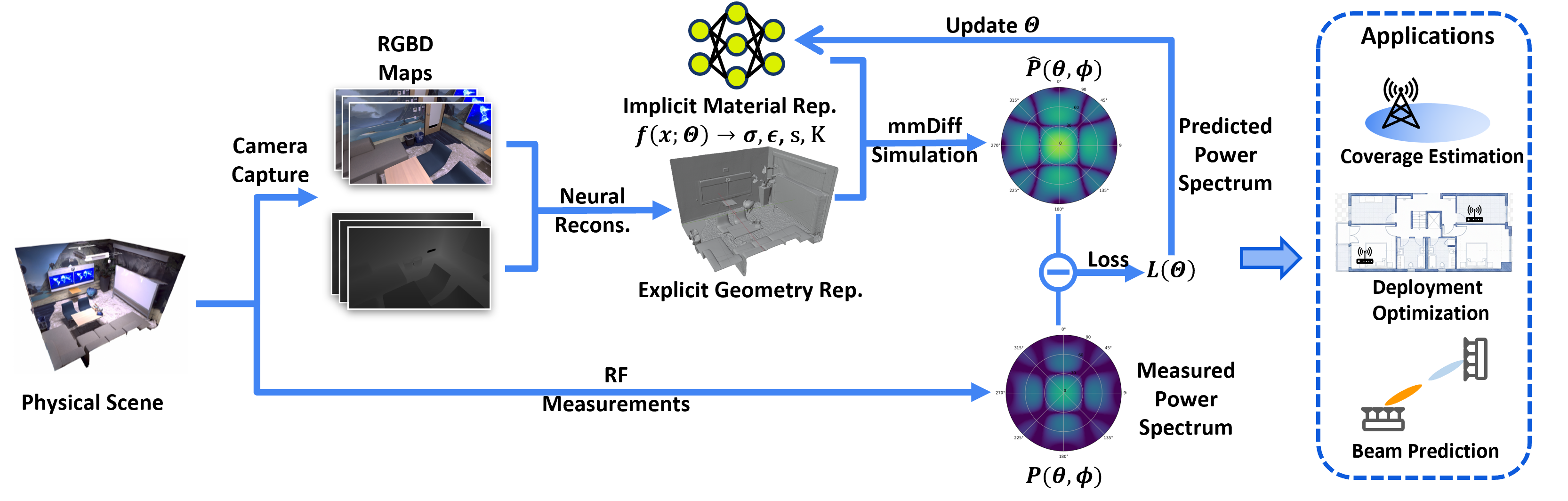}
    \caption{\name's learning framework for scene calibration. The inputs are the initial scene configuration and the AoA power spectrum measurements. The outputs are the calibrated scene configuration and the predicted AoA power spectrum.}
    \label{fig:learning_framework}
    \vspace{-15pt}
\end{figure*}

\subsection{Physics of Wireless Ray-Tracing}
Wireless ray tracing applies the ray approximation of electromagnetic (EM) propagation to predict how radio waves propagate through complex environments, incorporating wireless-specific physics such as antenna radiation patterns, array beamforming, polarization, phase evolution, and frequency-dependent material responses. Its goal is to compute the electromagnetic field received at a specific RX by modeling how rays depart the transmitter, interact with surfaces via reflection, diffraction, and scattering, and arrive at each receiver antenna element with the correct amplitude, phase, and polarization.

Following Sionna~\cite{sionna}, we adopt a physically grounded formulation for modeling these components. In particular, the total field arriving at the receiver depends on four factors: a) the radiation patterns of antennas; b) the beamforming weights of antenna arrays; c) the path gain of free-space propagation; and d) the change of signal amplitude, phase, and polarization due to the ray-surface interactions. We provide a brief formulation for some of these components below:

The radiation pattern of antennas is modeled as complex spherical functions:
\begin{equation}
    \label{eq:antenna-pattern}
    \mathbf{C}(\theta, \varphi) = C_\theta(\theta, \varphi) \hat{\boldsymbol{\theta}}+C_{\varphi}(\theta, \varphi) \hat{\boldsymbol{\varphi}},
\end{equation}

where $\hat{\boldsymbol{\theta}}$ and $\hat{\boldsymbol{\varphi}}$ are unit basis vectors of the spherical coordinate system; $C_\theta(\theta, \varphi)$ and $C_{\varphi}(\theta, \varphi)$ are complex values, modeling the elevation and azimuth pattern of the antenna, respectively. The pattern can be any physically valid spherical function satisfying energy conservation~\cite{balanis2016antenna}.

Under the far-field plane-wave assumption~\cite{sionna}, an antenna array can be treated as a single effective antenna with an aggregated radiation pattern. Rays are emitted from the array center, and beamforming gain is applied through the synthetic array weight:

\begin{equation}
    \label{eq:array_weights}
    \mathbf{W_{TX}}\left(\mathbf{k}\right)=\sum_{n=1}^N w_n \cdot e^{j \frac{2 \pi}{\lambda} \mathbf{d}_n \cdot \mathbf{k}},
\end{equation}

where $N$ is the number of TX antenna elements; $\mathbf{k}$ is the direction of departure vector; $w_n$ is the precoding weight of the $n$-th antenna element to steer the TX beam direction; $\mathbf{d}_n$ is the relative position of the $n$-th antenna element with respect to the center of the array; $\lambda$ is the signal wavelength.

The free-space propagation and field transformation due to ray-surface interaction can be generally represented by a transformation matrix $\mathbf{T}$ that collectively accounts for the change of amplitude, phase, and polarization during propagation. Hence, the E-field of a single ray $i$, reaching an RX antenna element, can be expressed as:

\begin{equation}
    \label{eq:rx_field_single_ray}
    \mathbf{E_i} = \frac{\lambda}{4\pi}\,
        \mathbf{C}_{\mathrm{T}}
        \,\mathbf{W}_{\mathrm{TX}}(\mathbf{k}_i^{\mathrm{AoD}})
        \,\mathbf{T}_{i}
        \mathbf{C}_{\mathrm{R}}^{\mathrm{H}},
\end{equation}

where $\frac{\lambda}{4\pi}$ is a scaling factor account for the antenna aperture; $\mathbf{C_T}$ and $\mathbf{C_R}$ are short for $\mathbf{C}_{\mathrm{T}}(\theta_{\mathrm{T}},\varphi_{\mathrm{T}})$ and $\mathbf{C}_{\mathrm{R}}(\theta_{\mathrm{R}},\varphi_{\mathrm{R}})$ modeling the radiation patterns of TX and RX antennas. The final wireless channel is the accumulation of E-fields along all paths to the receiver antenna.

\vspace{-5pt}
\section{\name Design}
\label{sec:design}
We consider a practical application scenario, shown in~\cref{fig:learning_framework} where a network operator surveys a deployment environment using an RGB-D camera to scan the scene and 3D reconstruction techniques to build a mesh model. Simultaneously, the operator measures received signal power at different angles and beam configurations using an antenna array to calibrate material parameters. These parameters are represented implicitly by a neural network that maps 3D ray-surface intersection coordinates to material properties such as conductivity and permittivity. Predicted AoA power spectra are compared with physical measurements to compute the loss. Since \name's ray-tracing pipeline is fully differentiable, the gradients of the loss can be back-propagated to update neural network parameters. This process is repeated until convergence, yielding a calibrated scene representation that closely matches the real environment. The calibrated model can then predict channels for arbitrary TX-RX locations and antenna patterns without retraining, as predictions are computed dynamically.

Next, we detail \name's ray-tracing algorithm and noise-robust ray-surface interaction model, the core components enabling accurate and differentiable channel prediction.

\subsection{mmDiff's Ray-Tracing Algorithm}
\label{sec:ray-tracing}
We introduce a pure Monte Carlo-based ray-tracing algorithm to compute the AoA power spectrum. It synthesizes received signal power at each RX antenna element by integrating contributions from many rays reaching the antenna. The algorithm comprises two tightly integrated components: path tracing, which finds signal propagation paths connecting transmitter and receiver, and power integration, which integrates EM power to produce the AoA power spectrum. We discuss each component in turn.

\noindent\textbf{Path Tracing} We adopt a backward ray tracing approach, where rays are launched from RX to TX to ensure signal power is captured from all receiving directions. For every ray emitting from the RX, we find its intersections with the environment mesh. For each ray emitted from the RX, we find its intersections with the environment mesh. At each intersection, we (i) perform a visibility test to the TX and, when unobstructed, spawn a deterministic ray toward the TX; and (ii) sample a next-hop direction according to the event type (reflection, diffusion, or refraction) and the surface scattering pattern. For highly directional scattering patterns, we use importance sampling to concentrate samples in high-gain directions, reducing simulation variance. We trace rays recursively until a preset maximum depth.

\noindent\textbf{Power Integration} We integrate the received signal power to each receiver antenna element. We apply an additional phase shift based on the receiving direction of the ray, $\mathbf{k}$, and relative position of the antenna element in the array, $\mathbf{d}_m$, where $m$ is the index of the antenna element: 

\begin{equation}
    \label{eq:array_weights_rx}
    \mathbf{W^m_{RX}}\left(\mathbf{k}\right)= e^{-j \frac{2 \pi}{\lambda} \mathbf{d}_m \cdot \mathbf{k}}.
\end{equation}

The Monte-Carlo integration equation for the received signal power at RX antenna element $m$ is:
\begin{equation}
\label{eq:mc_integration}
\begin{aligned}
P_m
  &= \int_{\Omega}
      \left|\frac{\lambda}{4\pi}\,
        \mathbf{C}_{\mathrm{T}}
        \,\mathbf{W}_{\mathrm{TX}}(\mathbf{k}^{\mathrm{AoD}})
        \,\mathbf{T}
        \,\mathbf{C}_{\mathrm{R}}
      \right|^{2}
       \,\mathbf{W}^{\,m}_{\mathrm{RX}}(\mathbf{k}^{\mathrm{AoA}})
      \,d\Omega \\[4pt]
  &\approx
      \sum_{i=1}^{L}\!
      \left|
        \frac{\lambda}{4\pi}\,
        \mathbf{C}_{\mathrm{T}}
        \,\mathbf{W}_{\mathrm{TX}}(\mathbf{k}_i^{\mathrm{AoD}})
        \,\mathbf{T}_{i}
        \mathbf{C}_{\mathrm{R}}^{\mathrm{H}}
        \right|^{2}
        \,\mathbf{W}^{\,m}_{\mathrm{RX}}(\mathbf{k}_i^{\mathrm{AoA}}),
\end{aligned}
\end{equation}

where $\Omega$ is the set of sampled angles; $L$ is the number of ray directions we sampled uniformly on the receiver's sphere.

Once we have the received signal power at each RX antenna element, we can convert it to the AoA power spectrum by weighted summation with the steering vector of the RX antenna array:

\begin{equation}
\label{eq:aoa_power_spectrum}
P(\theta, \varphi) = \frac{1}{\sqrt{M}}\sum_{m=1}^{M} P_m \cdot  e^{-j \frac{2 \pi}{\lambda} \mathbf{d}_m \cdot \mathbf{k}(\theta, \varphi)},
\end{equation}

where $M$ is the number of RX antenna elements. We compute $P(\theta, \varphi)$ for a mesh grid of angles $\theta$ and $\varphi$ covering the upper hemisphere to generate the AoA power spectrum shown in \cref{fig:aoa-spectrum}. 

\subsection{Noise-Robust Ray-Surface Interaction}
\label{sec:directional_scattering}

\begin{figure}[t]
    \centering
    \includegraphics[width=0.9\linewidth]{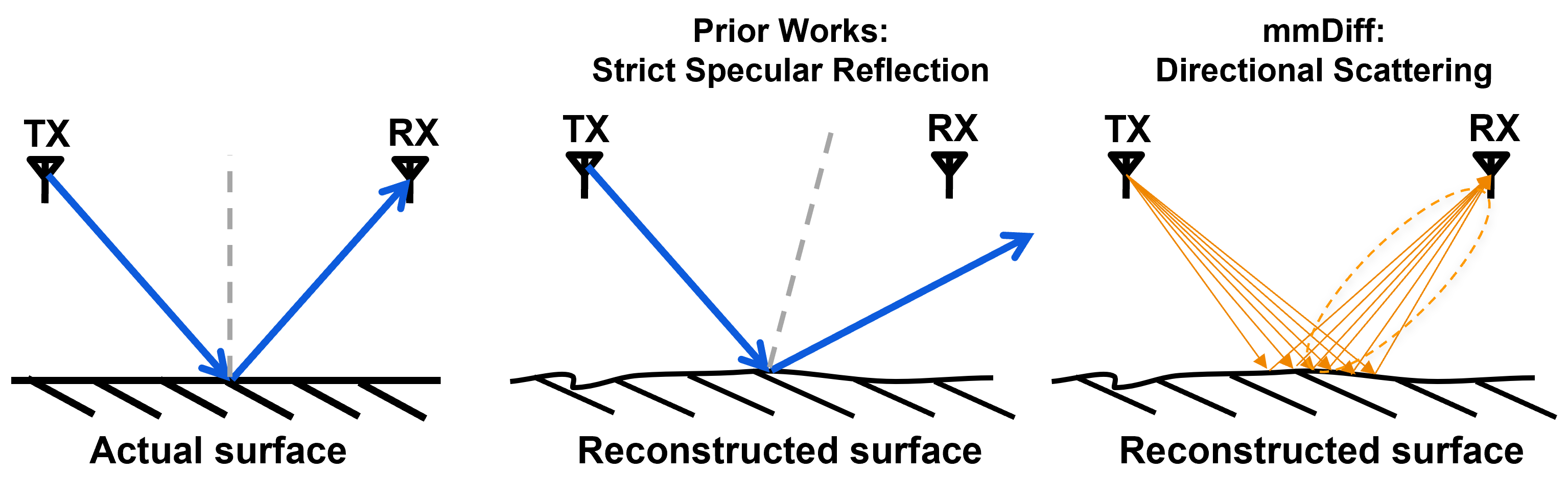}
    \caption{Issue with pure specular reflection: sensitivity to reconstruction noise, i.e., small discrepancies in the reconstruction model would cause significant differences in the simulated paths. 
    }
    \label{fig:spec_refl_issue}
    \vspace{-10pt}
\end{figure}

So far, we have explained \name's ray-tracing algorithm for generating the AoA power spectrum. Next, we discuss how we make the ray-tracing pipeline noise-robust and fully differentiable via a novel ray-surface interaction model. 

Four major types of interaction events are considered in common wireless simulators~\cite{sionna}: specular reflection, diffuse scattering, refraction, and diffraction. Diffraction is often ignored at mmWave frequencies due to its negligible contribution to overall path gain~\cite{carneiro2022study}. Specular reflection is the dominant propagation mechanism, contributing over 90\% of received power. However, specular reflection follows the ideal law of reflection, meaning incident and reflection angles are identical, so the outgoing direction is represented by a delta distribution. This makes ray tracing highly sensitive to geometric noise and non-differentiable with respect to surface normals. Even small perturbations in surface orientation can cause large deviations in reflection direction, leading to significant path discrepancies.

The sensitivity arises because the reflection lobe width depends on the ratio between the signal wavelength and the surface roughness~\cite{ticconi2011models}. For lower-frequency radios (sub-6 GHz), the wavelength is much larger than typical material roughness, so most surfaces appear mirror-like, with energy concentrated in the specular direction. At mmWave frequencies, however, the wavelength is comparable to the roughness of common building materials such as concrete, marble, and brick. Recent measurements confirm that reflections at mmWave exhibit narrow but nonzero angular spread rather than an ideal delta peak~\cite{guoMeasurementAnalysisScattering2025}. This motivates our directional scattering approximation, which replaces ideal specular reflection with a narrow-lobed scattering distribution. The resulting pattern preserves physical realism while maintaining differentiability of the ray-tracing pipeline with respect to geometry.

We adopt the directional scattering model from~\cite{degli2007measurement}, though any narrow-lobed formulation satisfying energy conservation applies. The scattering pattern is defined as:

\begin{equation}
    \label{eq:directional_scattering}
    f(\psi) \;=\;F_{\alpha_r}(\theta_i)^{-1}(\frac{1+\cos\psi}{2})^{\alpha_r},
\end{equation}

where $\psi$ is the deviation angle from the specular direction; $F_{\alpha_r}(\theta_i)^{-1}$ is a normalization factor that enforces energy conservation, depending on incident angle $\theta_i$; and $\alpha_r$ controls the lobe beamwidth, with larger values producing narrower, more directional patterns. Rather than relying on a single reflection point, directional scattering aggregates contributions from many nearby points to form a "heat spot", averaging reflected power and angles and making the model robust to geometric noise. Our implementation uses 64-bit arithmetic to support ultra-narrow beam computation. The half-power beamwidth can be as narrow as $1.40^\circ$.
\begin{figure}[t]                  
  \centering
  \begin{minipage}[t]{0.4\textwidth}
    \centering
    \includegraphics[width=\linewidth]{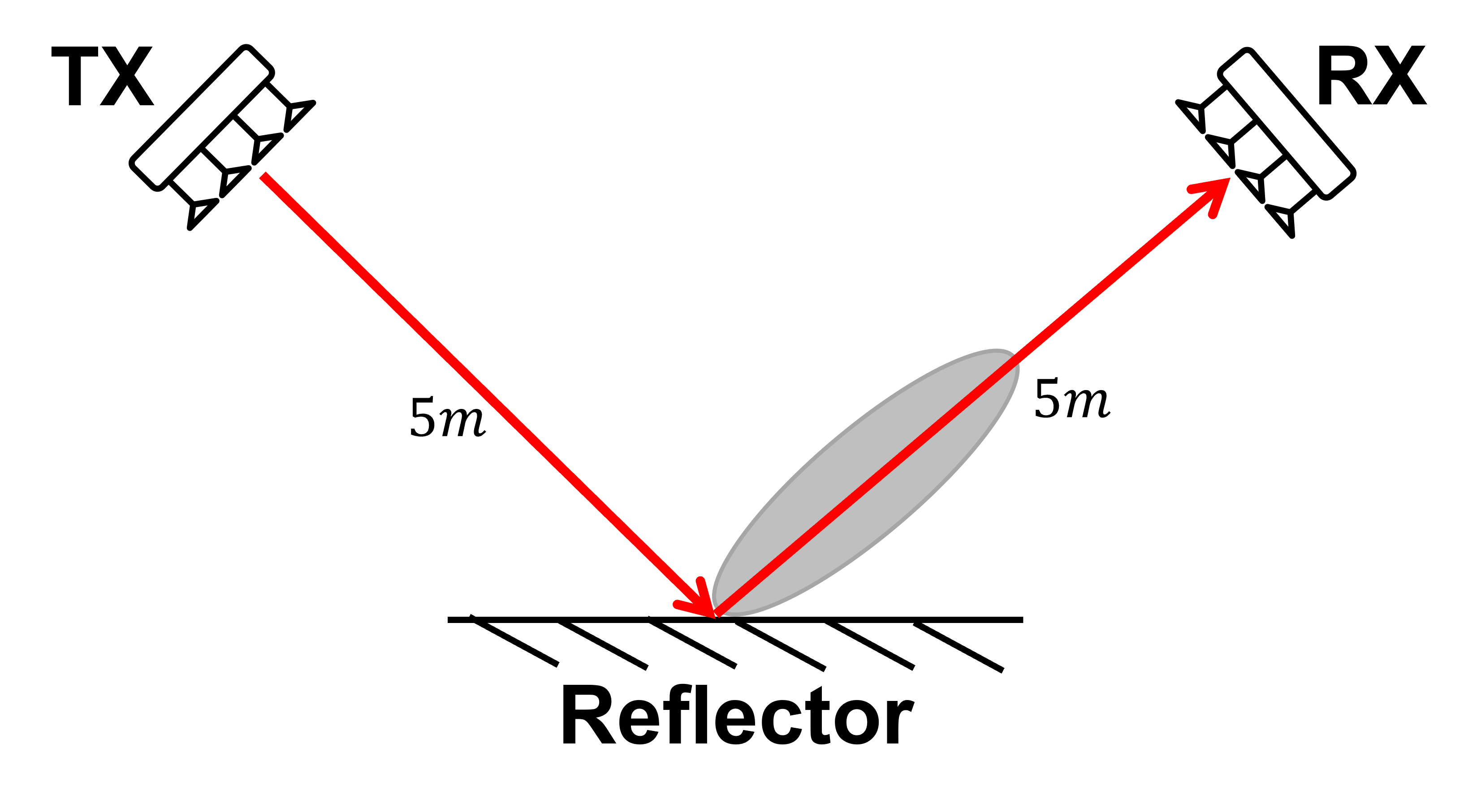}
    \captionof{figure}{Single-reflector scene for the directional scattering path gain convergence test.}
    \label{fig:single_reflector_scene}
  \end{minipage}
  \begin{minipage}[t]{0.4\textwidth}
    \centering
    \includegraphics[width=\linewidth]{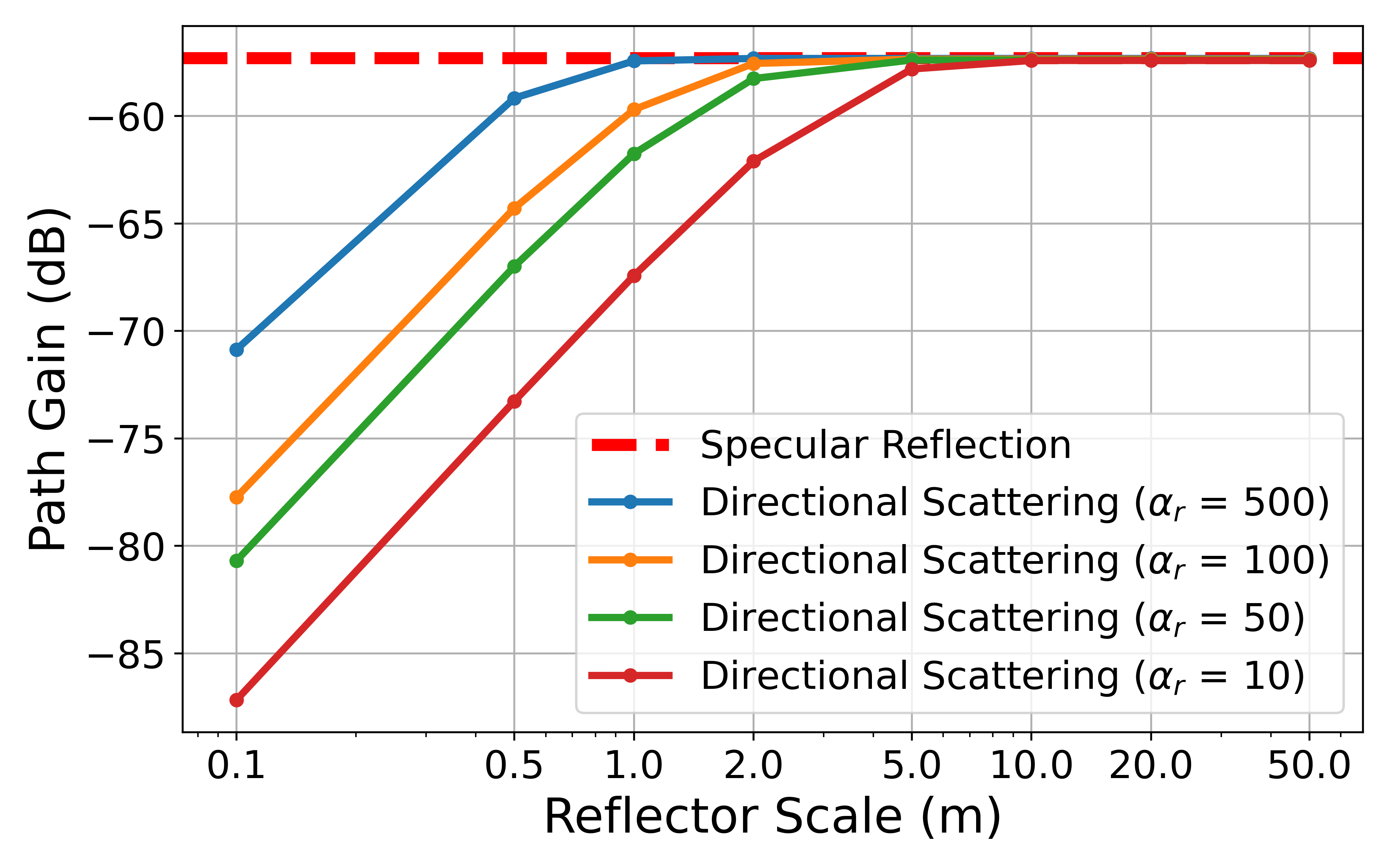}
    \captionof{figure}{Directional-scattering path gain converges to specular reflection path gain as the scattering footprint grows and directionality increases.}
    \label{fig:scat_converge}
  \end{minipage}
  \vspace{-10pt}
\end{figure}

Since we approximate specular reflection using directional scattering, we next establish a theorem showing that this approximation preserves the path gain:

\begin{theorem}
\label{theorem:scattering_convergence}
\textbf{Directional-Scattering Approximation:} For a planar reflector with scattering pattern $f_s(\psi)$, and reflection coefficient $\Gamma$, when the ratio of scattered field amplitude to incident field amplitude is given by the following equation: \footnote{This equation is simplified here by eliminating the cross-polar discrimination coefficient, for the sake of clarity. The actual implementation includes an additional transform matrix to account for the polarization.}
\begin{equation}
    \label{eq:scattering_path_coeff}
    \frac{|E_{scat}|}{|E_{inc}|} = \Gamma \sqrt{f_s(\psi)},
\end{equation}
The directional-scattering path gain asymptotically converges to the specular reflection path gain as the scattering footprint grows sufficiently large and the scattering lobe becomes sufficiently narrow.
\end{theorem}

We refer readers to Supplemental Material~\cref{app:scattering_convergence} for the full proof of this theorem. The intuition is that sufficiently large reflectors capture almost all power radiated from the transmitter's lower hemisphere. The captured power is then re-radiated towards the receiver, weighted by the scattering pattern. The highest weights are contributed by the scattering points near the specular reflection point. In the extreme case, when the scattering pattern collapses to a Dirac delta function, it matches the case of specular reflection. It is also worth noting that the classic Fresnel equations~\cite{balanis2016antenna}, which are widely used to compute specular reflection gain in wireless simulators, are derived under the assumption of an infinitely large planar interface. This assumption rarely holds in real environments with finite-sized reflectors. Our approximation smoothly bridges the transition from idealized infinite surfaces to realistic finite-scale reflectors, resulting in behavior that more closely matches real-world scenarios.

In practice, the scattering footprint does not have to grow infinitely large to reach a good approximation. The convergence rate depends on the directionality of the scattering pattern. We set up a simple experiment to demonstrate this behavior. As shown in \cref{fig:single_reflector_scene}, we put a pair of TX and RX in a scene with a single planar reflector. We put the TX and RX symmetrically on the same side of the reflector. We use Sionna's path solver to find the specular reflection path gain as a reference. We then use \name to simulate the gain of the directional scattering path with different reflector sizes and scattering pattern lobewidth. \cref{fig:scat_converge} shows the result. As shown, all curves converge to the specular path gain as the reflector scale increases. The more directional the scattering pattern, the faster the convergence. More analysis about the approximation behavior can be found in Supplemental Material~\cref{app:dev_analysis}.

\subsection{Optimization Framework}
\label{sec:optimization}
We use a fully-connected neural network with four hidden layers of 64 neurons each to parameterize the material properties in the reconstructed mesh. The neural network is evaluated at the 3D coordinates of each ray-surface intersection point (with 6-octave positional encoding~\cite{NeRF}) to predict the permittivity, conductivity, scattering coefficients, and cross-polar discrimination coefficients of the surface. The predicted material parameters are fed into physics-based equations to compute the scattered fields. The scattered fields are eventually integrated to compute the predicted AoA power spectra. We adopt the log mean absolute error (log-MAE) as the loss function for comparing the predicted and measured AoA power spectra, defined as:

\begin{equation}
    \label{eq:nmae}
    \mathrm{L_{log-MAE}} = \frac{1}{K}\sum_{i=1}^{K} |\log P_i^{\mathrm{pred}} - \log P_i^{\mathrm{meas}}|,
\end{equation}

where the AoA power spectrum is represented as a real-valued matrix of size $K = N_{el}\times N_{az}$, with $N_{az}$ and $N_{el}$ denoting the numbers of elevation and azimuth angles, respectively. We select log-MAE because it effectively captures errors across a wide dynamic range of power values, which is common in mmWave channels.

\subsection{Implementation}
We implement \name as a fork of Sionna v1.2.1~\cite{sionna}, which is built on Mitsuba v3.7.1~\cite{jakob2022mitsuba3} and Dr.Jit v1.2.0~\cite{Jakob2022DrJit}. We build a new stand-alone ray-tracing solver for synthesizing the AoA power spectrum, and a customized radio material implementation that integrates the directional scattering approximation and neural network material parameterization. We keep the original diffuse scattering unchanged but only replace the specular reflection with the directional scattering model. The scattering coefficient is now used to account for the power distribution between diffuse scattering and directional scattering. We test the simulator on an NVIDIA RTX 3090 GPU and observe a speed of 45 fps for training and 120 fps for inference with one million rays and up to two bounces, satisfying the real-time inference requirement. 

\vspace{-5pt}
\section{Results}
\label{sec:eval}
We conduct both real-world and simulation experiments to evaluate \name. 

\subsection{Real-world Experiments}
We constructed a real-world mmWave testbed in an indoor laboratory environment. A Keysight M9384B vector signal generator serves as the transmitter, producing a continuous-wave 28 GHz sinusoidal signal that is radiated via a Mi-Wave 261(34)-20/595 horn antenna. On the receiver side, an N9021B spectrum analyzer captures the incoming signal through an identical horn antenna, which is mechanically rotated in the azimuth plane to measure the angular power spectrum across eight equally spaced azimuth angles, as illustrated in \cref{fig:real_scene_setup}(a).

The 3D environment is captured using the PolyCam app on an iPad, from which RGB-D frames and camera trajectories are exported and processed via Neural RGB-D Reconstruction to generate a 3D mesh model of the scene (\cref{fig:real_scene_setup}(b)). Received signal power measurements are collected at 29 sparsely distributed locations, whose spatial distribution is shown in \cref{fig:real_scene_setup}(c), of which 23 locations (80\%) are reserved for training and the remaining 6 (20\%) for testing. The \name simulator is then configured to match the physical testbed parameters, including antenna radiation patterns, transmit power, and system losses. 

We train \name for 20 epochs. The directional scattering lobe parameter is set to $\alpha_r=300$, corresponding to a half-power beamwidth of 4.8$^\circ$. We use Sionna v0.19.1~'s path solver with line-of-sight (LoS) and specular reflection enabled as a baseline, and follow the procedure in~\cite{diff-rt-calibration} to perform calibration.

\cref{fig:real_results}(a) presents the predicted AoA power spectrum at a representative test position, where \name demonstrates substantially closer agreement with the ground truth compared to Sionna. \cref{fig:real_results}(b) shows the cumulative distribution function (CDF) of the absolute prediction error over the evaluation set (6 test positions $\times$ 8 angles), where \name achieves a mean absolute error of 4.60~dB, compared to 15.11~dB for Sionna — a reduction of more than 10~dB. The primary source of Sionna's degraded performance is noisy geometry reconstruction, which introduces errors in both path identification and power calibration. In contrast, \name's directional scattering model is inherently robust to geometric noise, enabling reliable recovery of dominant propagation paths and yielding significantly more accurate power predictions. These results validate the practical effectiveness of \name in real-world deployment scenarios, where geometrically perfect reconstructions are rarely attainable.

\begin{figure*}[t]
    \centering
    \hfill
    \begin{minipage}[t]{0.55\textwidth}
        \centering
        \includegraphics[width=\linewidth]{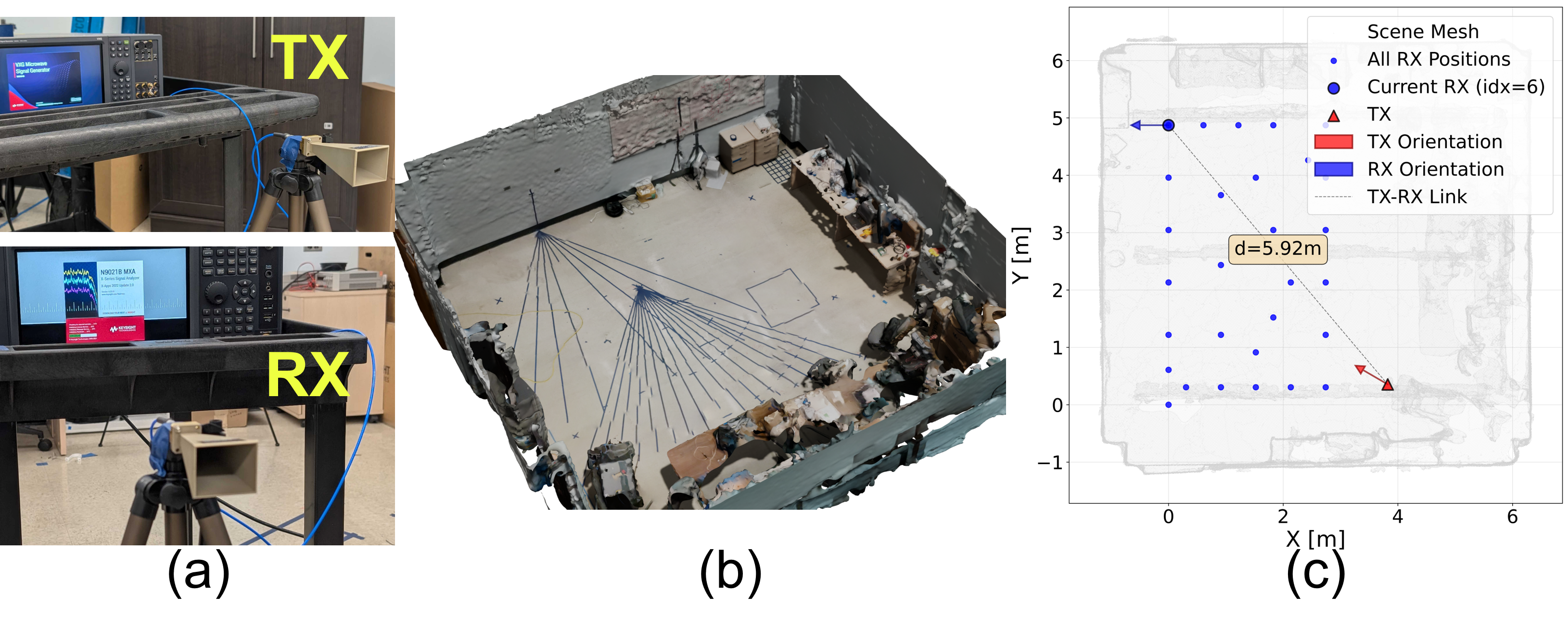}
        \caption{Real-world testbed setup.}
        \label{fig:real_scene_setup}
    \end{minipage}
    \begin{minipage}[t]{0.43\textwidth}
        \centering
        \includegraphics[width=\linewidth]{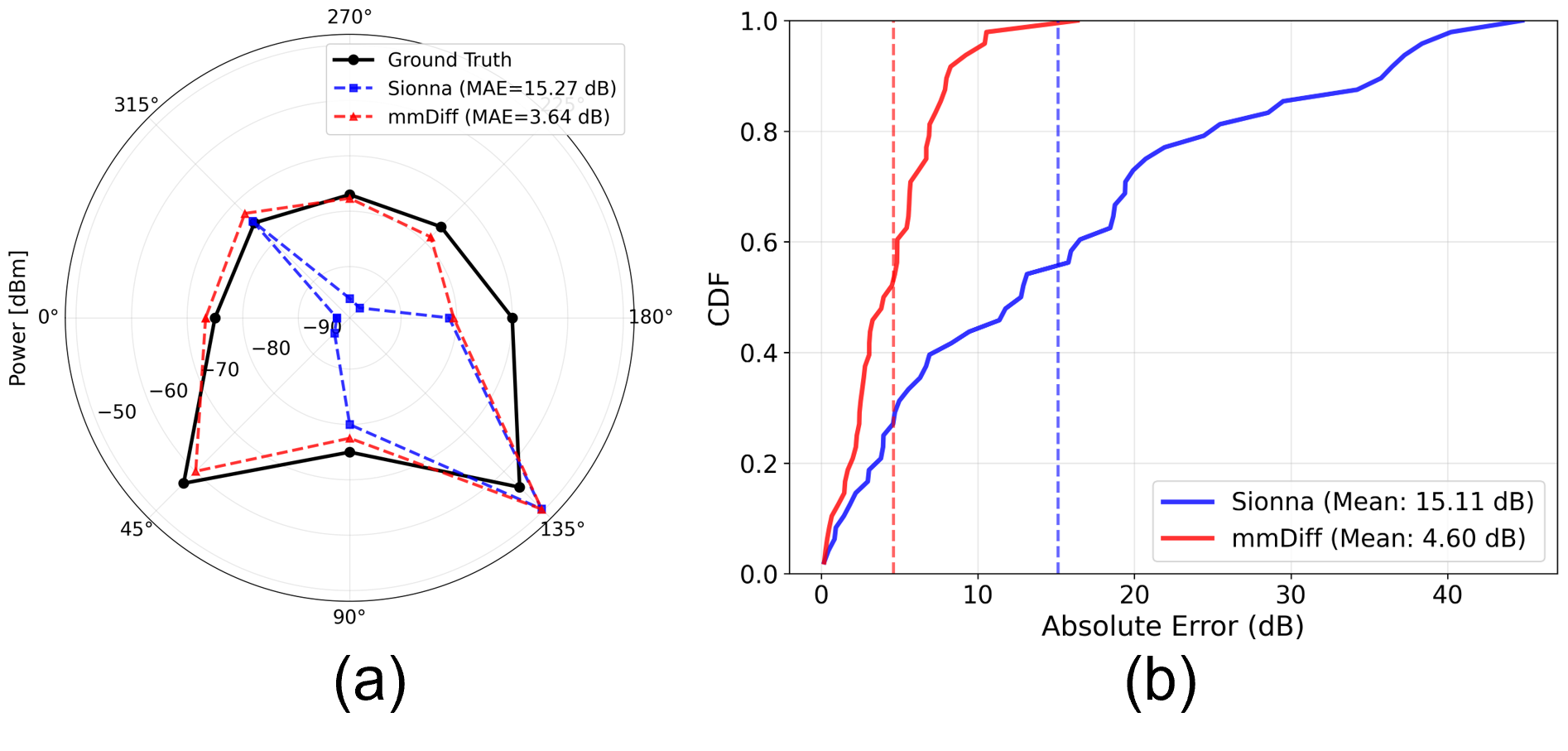}
        \caption{Results of the real-world experiment.}
        \label{fig:real_results}
    \end{minipage}
    \hfill
\end{figure*}

\subsection{Simulation Experiments}
\label{sec:dataset}
\begin{figure}
    \centering
    \includegraphics[width=0.9\columnwidth]{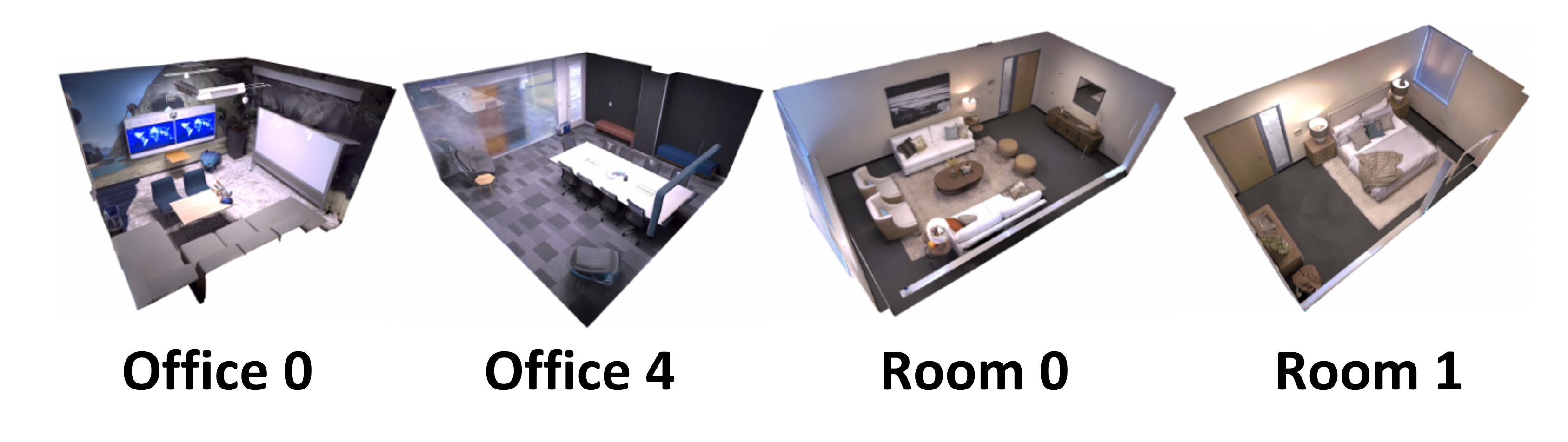}
    \caption{Replica~\cite{replica19arxiv} scenes used in evaluation.}
    \label{fig:replica_scenes}
    \vspace{-15pt}
\end{figure}
\noindent\textbf{Datasets} We evaluate \name on four realistic 3D indoor scenes from the Replica dataset~\cite{replica19arxiv} (\cref{fig:replica_scenes}), using the original meshes as ground truth (GT) with ITU materials~\cite{series2015effects} assigned based on semantic labels. For each scene, we create a camera trajectory in Blender and render RGB-D images using the Replica Renderer; each trajectory contains 2800 poses covering the majority of the scene area. These poses also serve as receiver locations in a wireless simulator to generate the corresponding mmWave channel datasets (TRAJ). We apply Neural RGB-D Reconstruction~\cite{nrgbd} to recover scene geometry from the RGB-D images, producing reconstructed meshes (REC), and manually remove floaters introduced by the algorithm. This process introduces reconstruction noise similar to real-world captures. We further sample ${\sim}500$ random receiver poses per scene as a held-out test set (RAND) to probe diverse viewpoints. For all datasets, the TX is fixed at a single pose. Both transmitter and receiver use a $4{\times}4$ antenna array; simulations run at 28~GHz with one million rays traced to a maximum depth of one. We empirically choose $\alpha_r=100$ for all simulation experiments. Further dataset details appear in Supplemental Material~\cref{app:dataset_details}.

\noindent\textbf{Baselines} We use Sionna v0.19.1~'s path solver with line-of-sight (LoS) and specular reflection enabled as a baseline. Sionna traces diffuse scattering paths; however, these paths are used for predicting the Channel Impulse Response (CIR), rather than approximating the specular path power. Hence, we disabled diffuse scattering for Sionna in our experiments. Sionna does not directly produce AoA spectra. We apply the same beamforming procedure described in ~\cref{eq:mc_integration} and~\cref{eq:aoa_power_spectrum} to obtain comparable AoA spectra from Sionna's traced paths. 

\noindent\textbf{Experiment~1: Robustness to Geometry Noise}

\begin{figure}[!ht]
    \centering
    \includegraphics[width=0.8\columnwidth]{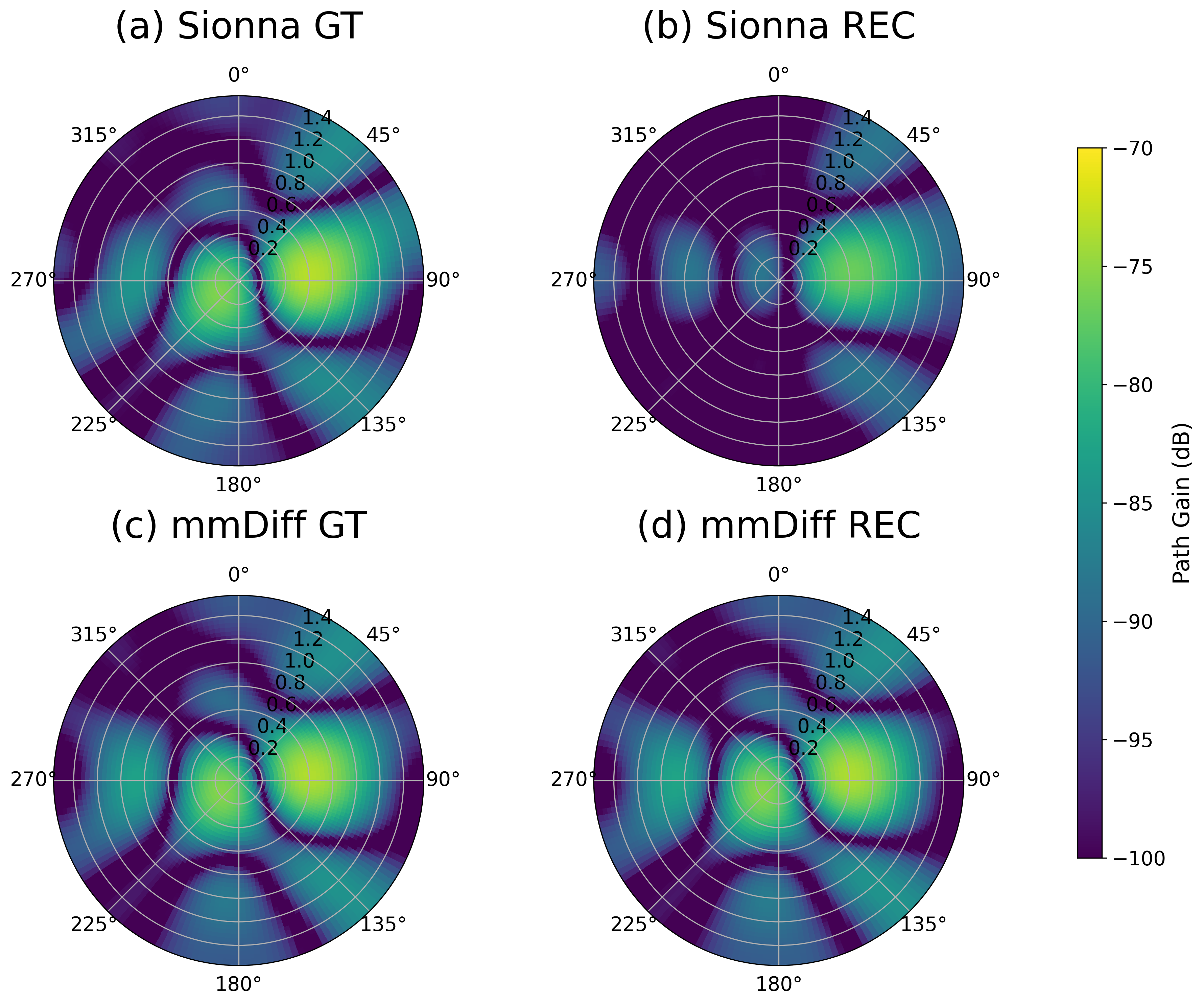}
    \caption{Sionna exhibits significant deviation in the REC environment while \name performs consistently despite geometry noise.}
    \vspace{-5pt}
    \label{fig:robustness}
\end{figure}

\begin{table}[t]
\centering
\caption{Robustness under Geometry Noise (TRAJ Dataset)}
\setlength{\tabcolsep}{2pt}
\renewcommand{\arraystretch}{0.95}
\small
\begin{tabular}{ll|ccc|ccc}
\hline
Scene & Methods 
& \multicolumn{3}{c|}{Top-1 Peak} 
& \multicolumn{3}{c}{Top-5 Peaks} \\ 
& & F1$\uparrow$ & AAE$\downarrow$ & APE$\downarrow$
  & F1$\uparrow$ & AAE$\downarrow$ & APE$\downarrow$ \\
\hline
\multirow{2}{*}{Office 0}
 & Sionna  & 0.8046 & 3.15 & 1.27 & 0.7431 & 3.30 & 1.93 \\ 
 & \textbf{\name} & \textbf{0.9954} & \textbf{0.29} & \textbf{0.14}
          & \textbf{0.9804} & \textbf{0.38} & \textbf{0.18} \\ 
\hline
\multirow{2}{*}{Office 4}
 & Sionna  & 0.9546 & 2.00 & 0.79 & 0.8886 & 1.74 & 0.98 \\ 
 & \textbf{\name} & \textbf{0.9950} & \textbf{0.52} & \textbf{0.11}
          & \textbf{0.9774} & \textbf{0.62} & \textbf{0.16} \\ 
\hline
\multirow{2}{*}{Room 0}
 & Sionna  & 0.9707 & 1.71 & 0.97 & 0.9133 & 1.64 & 1.02 \\ 
 & \textbf{\name} & \textbf{0.9936} & \textbf{1.33} & \textbf{0.55}
          & \textbf{0.9751} & \textbf{1.07} & \textbf{0.62} \\ 
\hline
\multirow{2}{*}{Room 1}
 & Sionna  & 0.8086 & 3.89 & 1.75 & 0.7710 & 3.14 & 2.16 \\ 
 & \textbf{\name} & \textbf{0.9925} & \textbf{0.26} & \textbf{0.07}
          & \textbf{0.9879} & \textbf{0.28} & \textbf{0.10} \\ 
\hline
\end{tabular}
\label{tab:robustness}
\vspace{-10pt}
\end{table}
We first evaluate the robustness of \name to geometry noise that is commonly observed in 3D reconstruction.

In this experiment, we run both simulators on the ground-truth (GT) and reconstructed (REC) meshes across all four Replica scenes. To exclude the effect of material estimation, we assign the same uniform material to all scene objects. We generate AoA power spectra on REC and GT meshes along the same TRAJ datasets.
\cref{fig:robustness} illustrates the results for one of the receiver poses in Office 0. As is shown, on the GT geometry, \name and Sionna produce similar spectra; however, on the REC geometry, \name closely matches the GT results while Sionna exhibits substantial deviations.

For quantitative evaluation, we detect dominant peaks in the AoA power spectra using a maximum-filter-based peak detection algorithm. We extract the top-1 and top-5 peaks and match them between the REC and GT spectra based on angular proximity (within 20$^\circ$). We report the F1-score for peak detection accuracy and compute the Average Angle Error (AAE) in degrees and Average Power Error (APE) in dB for matched peaks. We choose these metrics as they directly reflect the accuracy of angle and power estimation for practical applications such as beam prediction.
\cref{tab:robustness} summarizes the results: across all scenes and metrics, \name exhibits almost perfect peak detection accuracy and only minor angle and power error, consistently outperforming the Sionna baseline. This demonstrates the robustness of \name against geometry noise, a critical feature for real-world applications.

\noindent\textbf{Experiment 2: Scene Calibration and Channel Prediction}
\begin{table}[t]
\centering
\caption{Prediction Accuracy on RAND Dataset}
\setlength{\tabcolsep}{2pt}
\renewcommand{\arraystretch}{0.95}
\footnotesize
\begin{tabular}{ll|ccc|ccc}
\hline
Scene & Methods
& \multicolumn{3}{c|}{Top-1 Peak}
& \multicolumn{3}{c}{Top-5 Peaks} \\
& & F1$\uparrow$ & AAE$\downarrow$ & APE$\downarrow$
  & F1$\uparrow$ & AAE$\downarrow$ & APE$\downarrow$ \\
\hline
\multirow{3}{*}{Office 0}
 & NeRF\textsuperscript{2}  & 0.0329 & 12.86 & 8.96 & 0.2082 & 13.50 & 8.51 \\
 & Sionna & 0.6708 & 5.03 & 3.94 & 0.6794 & 3.88 & 4.52 \\
 & \name w/o Calib. & 0.6249 & 4.56 & 3.98 & 0.6845 & 4.02 & 4.97 \\
 & \textbf{\name Full} & \textbf{0.9168} & \textbf{1.84} & \textbf{1.08}
          & \textbf{0.9121} & \textbf{1.75} & \textbf{1.22} \\
\hline
\multirow{3}{*}{Office 4}
 & NeRF\textsuperscript{2}  & 0.0460 & 14.28 & 3.15 & 0.0816 & 12.74 & 4.06 \\
 & Sionna & 0.8154 & 4.16 & 2.03 & 0.7554 & 3.47 & 2.42 \\
 & \name w/o Calib. & 0.5643 & 8.16 & 7.15 & 0.6197 & 5.15 & 7.44 \\
 & \textbf{\name Full} & \textbf{0.9257} & \textbf{2.08} & \textbf{0.97}
          & \textbf{0.8994} & \textbf{1.58} & \textbf{1.11} \\
\hline
\multirow{3}{*}{Room 0}
 & NeRF\textsuperscript{2}  & 0.2611 & 12.85 & 0.91 & 0.3822 & 9.38 & 3.32 \\
 & Sionna & 0.9087 & 3.16 & 1.94 & 0.8345 & 2.71 & 2.11 \\
 & \name w/o Calib. & 0.6505 & 6.80 & 2.92 & 0.6948 & 4.71 & 3.65 \\
 & \textbf{\name Full} & \textbf{0.9534} & \textbf{1.36} & \textbf{0.56}
          & \textbf{0.9441} & \textbf{1.26} & \textbf{0.64} \\
\hline
\multirow{3}{*}{Room 1}
 & NeRF\textsuperscript{2}  & 0.1868 & 12.29 & 2.43 & 0.3186 & 9.41 & 4.14 \\
 & Sionna & 0.7859 & 5.39 & 4.34 & 0.7422 & 4.15 & 4.34 \\
 & \name w/o Calib. & 0.7473 & 4.59 & 2.10 & 0.7537 & 3.86 & 2.85 \\
 & \textbf{\name Full} & \textbf{0.9342} & \textbf{1.58} & \textbf{0.92}
          & \textbf{0.9287} & \textbf{1.43} & \textbf{0.95} \\
\hline
\end{tabular}
\label{tab:ch_prediction}
\vspace{-15pt}
\end{table}
We run end-to-end experiments to validate that \name's robustness and differentiability benefits the material calibration and channel prediction. For Sionna and \name, we separately generate synthetic datasets on the GT mesh as ground truth channels for supervision. During training, we use the ray-surface interaction points traced in the REC meshes to query the neural material model. We include another data-driven baseline approach, NeRF\textsuperscript{2}~\cite{NeRF2}, for this channel prediction task, which adopts NeRF-like ray tracing to learn the radiation fields of a wireless environment. The synthetic AoA power spectra generated by Sionna are used for training and testing of NeRF\textsuperscript{2} models. We train all models using poses along a designated trajectory (TRAJ) and test on poses randomly sampled in the scene (RAND). More details about the training procedure can be found in Supplemental Material ~\cref{app:train_details}.

The results of this experiment are shown in~\cref{tab:ch_prediction}: NeRF\textsuperscript{2} fails to generalize under such sparse sampling density due to the lack of explicit geometric priors; \name consistently outperforms both baselines across all scenes and metrics, achieving the highest accuracy in mmWave channel prediction. As an ablation study, we also report results for \name without calibration, which perform poorly and underscore both the necessity of material calibration and the effectiveness of \name. We also note that the overall performance on the RAND dataset is slightly lower than in \cref{tab:robustness}, as it includes more diverse receiver locations where geometric reconstruction is imperfect due to occlusions. Nevertheless, \name maintains high prediction accuracy, sufficient for practical deployment scenarios.

\noindent\textbf{Additional Results:} We include more experiments and results in the Supplemental Material~\cref{app:add_results}. 

\vspace{-5pt}
\section{Discussion}
\label{sec:limitation}
\begin{itemize}

    \item \textbf{Geometry optimization:} \name is differentiable with respect to the scene geometry, which has the potential application of calibrating scene geometry. We demonstrate this capability in Supplemental Material~\cref{app:geo_diff}, where we calibrate the translation and rotation of a planar reflector. However, calibrating the full geometry mesh is substantially harder and beyond this work's scope; we believe this is an interesting direction for future exploration.

    \item \textbf{Dynamic scenarios:} \name extends naturally to dynamic environments. In practice, large-scale reflectors, such as walls, ceilings, and building facades, are static and require no retraining as the environment evolves. For dynamic objects such as humans or moving furniture, visual tracking (e.g., from RGB-D cameras or LiDAR) can continuously update their positions and geometry within the scene representation. Leveraging \name's fast inference, these position updates translate directly into real-time channel predictions without re-running the full simulation pipeline, making \name well-suited for dynamic scenarios.

    \item \textbf{Phase Errors:} \name integrates received signal power over directions via Monte Carlo sampling, discarding per-ray phase to ensure accurate power integration. This may produce inconsistent interference patterns when multiple paths carry comparable strength, as we show in Supplemental Material~\cref{app:multipath}. However, in most real-world scenarios, mmWave transceivers are using very narrow beams, and therefore, there is only a single path dominating power compared to the other paths. Hence, the impact of phase error is negligible, and the calibration is still functional.
\end{itemize}

\vspace{-7pt}
\section{Conclusion}
\vspace{-7pt}
\label{sec:conclusion}
In conclusion, we propose the directional scattering approximation of specular reflection to handle the inevitable geometry imperfections in 3D reconstructed models, enabling robust and differentiable simulation of mmWave channels. We build \name, the first noise-robust, fully differentiable framework for mmWave scene calibration and channel prediction. Leveraging end-to-end differentiability, \name learns scene-specific material parameters in geometrically complex environments through gradient-based optimization, eliminating the need for manual parameter tuning. Overall, these properties position \name as a practical and deployable solution for real-world mmWave channel modeling, with applications spanning base station coverage prediction, network deployment planning, and beam training overhead reduction in next-generation wireless networks. Looking ahead, promising directions for future work include scaling \name to larger outdoor environments, extending to dynamic scenes with continuous online learning, and integrating mmWave sensing capabilities to enable joint communication and sensing.

\noindent\textbf{Acknowledgements:}
We thank the UCLA ICON group and the reviewers for their insightful comments. This work is supported by UCLA, and NSF awards 2238245.

{
    \small
    \bibliographystyle{ieeenat_fullname}
    \bibliography{main}
}

\newpage

\newpage
\appendix
\onecolumn

\begin{center}
\LARGE mmDiff: A Noise-Robust Differentiable Ray-Tracing Framework for mmWave Scene Calibration and Channel Prediction \\
{\Large\normalfont Supplementary Material}
\end{center}

\section{Proof of Theorem 1}
\label{app:scattering_convergence}
\begin{figure}[h]
    \centering
    \hfill
    \begin{minipage}[b]{0.45\linewidth}
        \centering
        \includegraphics[width=\linewidth]{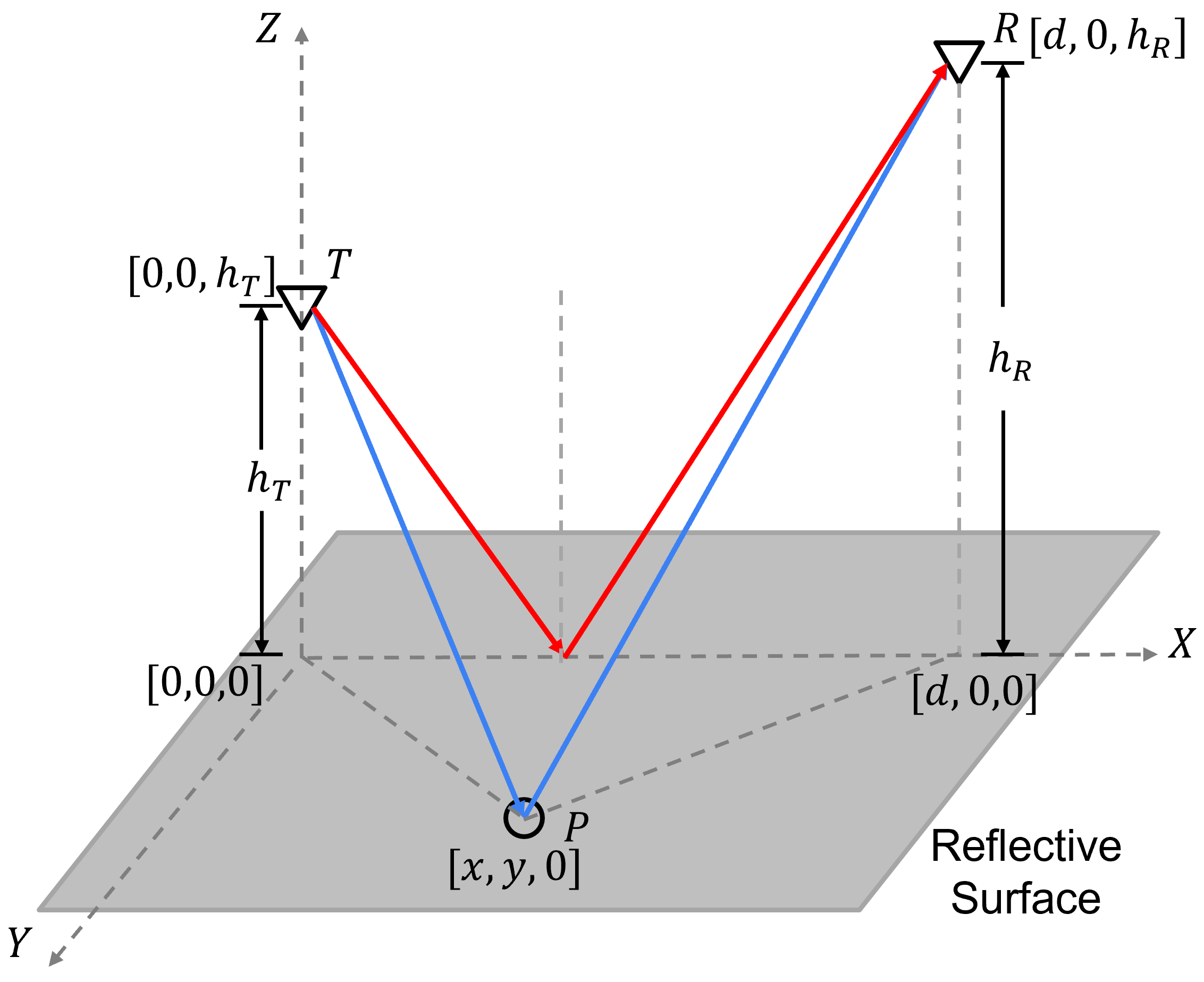}
        \caption{Geometry setting for proof of Theorem 1.}
        \label{fig:proof_geo}
    \end{minipage}
    \hfill
    \begin{minipage}[b]{0.45\linewidth}
        \centering
        \includegraphics[width=\linewidth]{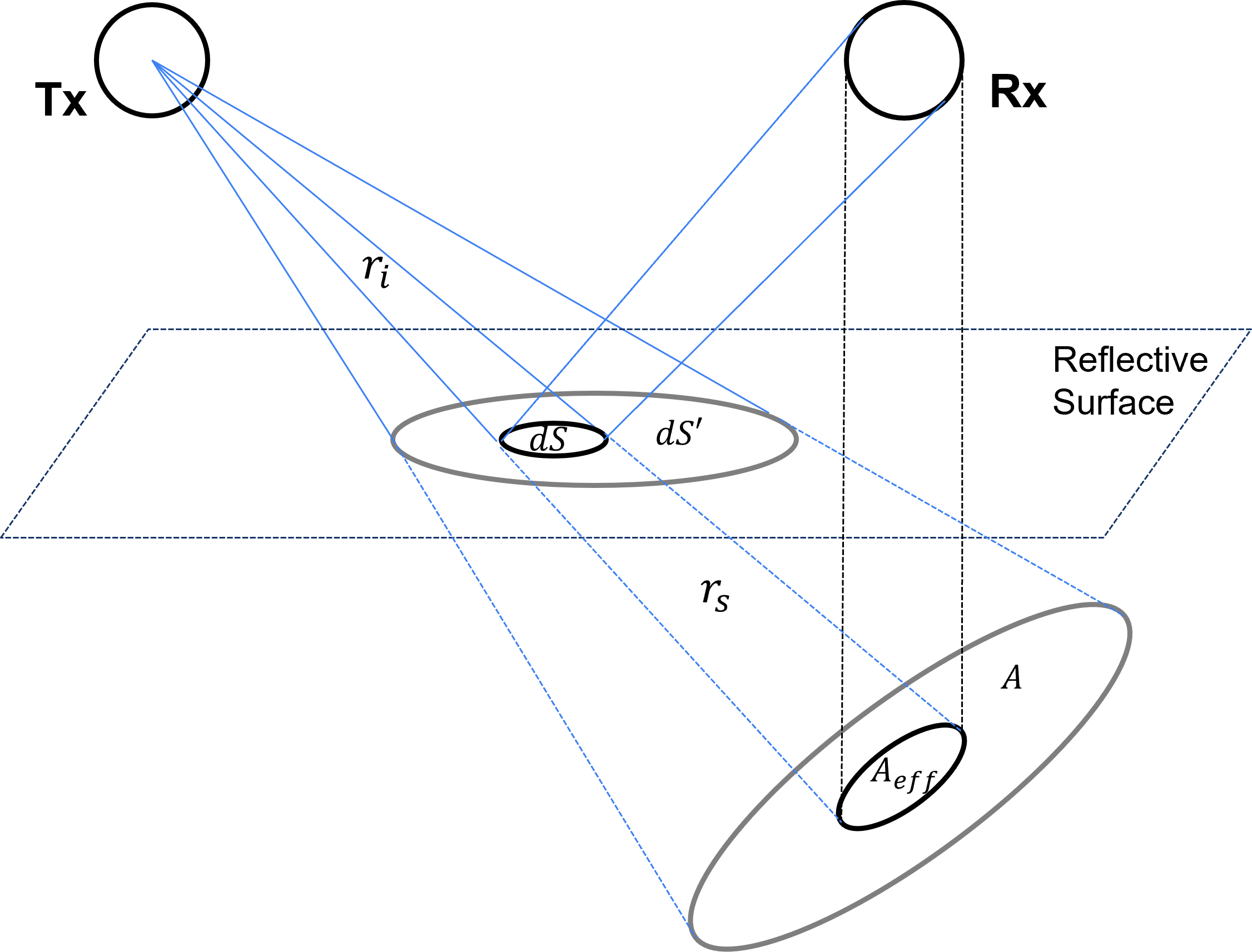}
        \caption{Illustration of the effective area $dS$.}
        \label{fig:proof_effective_area}
    \end{minipage}
    \hfill
\end{figure}

We prove the following equation holds:
\begin{equation}
    \label{eq:scattering_convergence}
    \lim_{\alpha_r \rightarrow \infty}\lim _{S \rightarrow \infty} P_{\alpha_r}(S)=\left(\frac{\lambda \Gamma}{4 \pi r_{\text {spec }}}\right)^2,
\end{equation}
where $P_{\alpha_r}(S)$ denotes the power received at RX that is carried by single-bounce diffuse-scattering paths whose scattering point lies in area $S$, We show that when the reflector surface is sufficiently large, the received scattering power approaches the specular reflection power asymptotically as the scattering pattern becomes more directional. 

We start with an ideal, vertically polarized, isotropic antenna model and will generalize to an arbitrary antenna pattern towards the end. 
\begin{proof}
Since the received power will always be a fraction of the transmitted power, for simplicity, we assume the transmitted power $P_t = 1$. 

For an infinitesimal patch of the surface, $dS'$, the solid angle subtended by the patch at the TX is 
\begin{equation}
    d \Omega_i=\frac{\cos \theta_i d S'}{r_i^2},
\end{equation}
where $\theta_i$ is the incident angle to the scattering point.

Hence, the differential incident power is 
\begin{equation}
    d P_{\mathrm{inc}}=\frac{1}{4 \pi} d \Omega_i=\frac{\cos \theta_i}{4 \pi r_i^2} d S'
\end{equation}

Since the scattered field amplitude is $\Gamma \sqrt{f_s (\alpha_r; \theta_s, \phi_s)}$ times of the incident field (\cref{eq:scattering_path_coeff}), the differential scattered power is 

\begin{equation}
    d P_{\text {sca }}=\Gamma^2 f_s(\alpha_r; \theta_s, \phi_s) d P_{\text {inc }}
\end{equation}

The scattering field propagates to RX in free space, hence its power decay follows the Friis equation of free-space path loss. 

\begin{equation}
    d P_r = (\frac{\lambda}{4\pi r_s})^2 dP_{\text {sca }},
\end{equation}

Thus, the differential received power is:
\begin{equation}
    \label{eq:diff_power_contribution}
    d P_r=\left(\frac{\lambda \Gamma}{4 \pi}\right)^2 f_s(\alpha_r; \theta_s, \phi_s) \frac{\cos \theta_i}{r_i^2 r_s^2} d S' 
\end{equation}

To get the total received power, we need to integrate over the whole reflective surface:
\begin{equation}
    \label{eq:int_rx_power}
    P_{\alpha_r}(S)=\iint_S\left(\frac{\lambda \Gamma}{4 \pi}\right)^2 f_s(\alpha_r; \theta_s, \phi_s) \frac{\cos \theta_i}{r_i^2 r_s^2} d S' 
\end{equation}

It is essential to note that only a fraction of the total shined area is captured by the receiver's aperture. As illustrated in \cref{fig:proof_geo}, the total illuminated area by a ray is denoted as $dS'$. The total power that reaches $dS'$ is scattered (re-distributed) to different directions, only a fraction of which is captured by the effective aperture of the receiver, denoted as $A_{eff}$. According to the geometry, the effective area, $dS$, that contains the power that reaches the receiver is related to $dS$ by the following equation:

\begin{equation}
    \frac{dS}{dS'} = \frac{A_{eff}}{A} = \frac{r_i^2}{(r_i+r_s)^2},
\end{equation}

where $A$ is an auxiliary variable denoting the subtended area of the ray at distance $r_i + r_s$.

Plug $dS$ into the power integration, we get:

\begin{equation}
    \label{eq:int_rx_power_effective_area}
    P_{\alpha_r}(S)=\iint_S\left(\frac{\lambda \Gamma}{4 \pi}\right)^2 f_s(\alpha_r; \theta_s, \phi_s) \frac{\cos \theta_i}{r_i^2 r_s^2} \frac{r_i^2}{(r_i+r_s)^2}dS
\end{equation}

To calculate the integral, we perform a change of variable to spherical coordinates $(\theta_s, \phi_s)$ at the scattering point. We use the geometry model in \cref{fig:proof_geo} to find the Jacobian matrix, where TX and RX are located at points $\mathbf{T} = (-0, 0, h_T), \mathbf{R} = (d, 0, h_R)$, respectively; $\mathbf{P}= (x, y, 0)$ represents an arbitrary scattering point on the reflector plane. Since $dS = dx dy$, we want to change the basis from $(x, y)$ to $(\theta_s, \phi_s)$. According to the geometry,

\begin{equation}
    \mathbf{k_s} = \begin{bmatrix}
    \sin\theta_s \cos\phi_s \\
    \sin\theta_s \sin\phi_s \\
    \cos\theta_s
    \end{bmatrix} = \frac{1}{r_s}\begin{bmatrix}
    d-x \\
    -y \\
    h_R
    \end{bmatrix},
\end{equation}

where $r_s = \sqrt{(d-x)^2 + y^2 + h_R^2}$. And we can solve the determinant of the Jacobian:

\begin{equation}
    |J| = \det(A) = \begin{vmatrix}
    \frac{\partial x}{\partial \theta_s} & \frac{\partial x}{\partial \phi_s} \\
    \frac{\partial y}{\partial \theta_s} & \frac{\partial y}{\partial \phi_s}
    \end{vmatrix} = r_s^2\tan\theta_s
\end{equation}
Apply the change-of-variable:
\begin{align}
 \lim _{S \rightarrow \infty}P_{\alpha_r}(S) 
&= \int_{0}^{2\pi} \int_{0}^{\pi/2}
   \left(\frac{\lambda \Gamma}{4 \pi}\right)^2 
   f_s(\alpha_r; \theta_s, \phi_s) 
   \frac{\cos \theta_i}{r_i^2 r_s^2} 
   \frac{r_i^2}{(r_i+r_s)^2} 
   r_s^2 \tan\theta_s \, d\theta_s \, d\phi_s \nonumber \\
&= \left(\frac{\lambda \Gamma}{4 \pi}\right)^2 
   \int_{0}^{2\pi} \int_{0}^{\pi/2}
   f_s(\alpha_r; \theta_s, \phi_s) 
   \sin\theta_s \,
   \frac{1}{(r_i+r_s)^2} 
   \frac{\cos \theta_i}{\cos\theta_s} 
   d\theta_s \, d\phi_s
\end{align}
The integral range includes the entire upper hemisphere to adhere to the limit of a sufficiently large surface.
According to the energy conservation (EC) property of the scattering pattern:

\begin{equation}
    \label{eq:energy_cons_scattering}
     \!\int_{0}^{2\pi}\int_{0}^{\pi/2}\!
    f_{s}(\alpha_r; \theta_s, \phi_s)
    \,\sin\theta_sd\theta_s\,d\phi_s\, \;=\;1.
    \tag{EC}
\end{equation}
As $\alpha_r\rightarrow\infty$, the scattering pattern approaches a Dirac Delta function in 3D, and $\cos\theta_i=\cos\theta_s$ at the specular reflection point. Thus, according to the sifting property of the Delta function, we have:

\begin{equation}
    \lim_{\alpha_r\rightarrow\infty}\lim _{S \rightarrow \infty}P_{\alpha_r}(S) =  \left(\frac{\lambda \Gamma}{4 \pi}\right)^2 \frac{1}{(r^*_i+r^*_s)^2} =  \left(\frac{\lambda \Gamma}{4 \pi r_{spec}}\right)^2 ,
\end{equation}
where $r^*_i$ and $r^*_s$ are the incident and scattering segment lengths when the scattering point matches the specular reflection point.
\end{proof}

The above proof assumes a vertically-polarized isotropic antenna at both TX and RX. Since we model the antenna radiation pattern and polarization as a separated multiplicative factor (\cref{sec:ray-tracing}), the above proof can be easily generalized to arbitrary antenna patterns and polarization. 

\newpage
\section{Deviation Analysis}
\label{app:dev_analysis}

\begin{figure*}[t]
    \centering
    \begin{subfigure}[t]{0.32\textwidth}
        \centering
        \includegraphics[width=\textwidth]{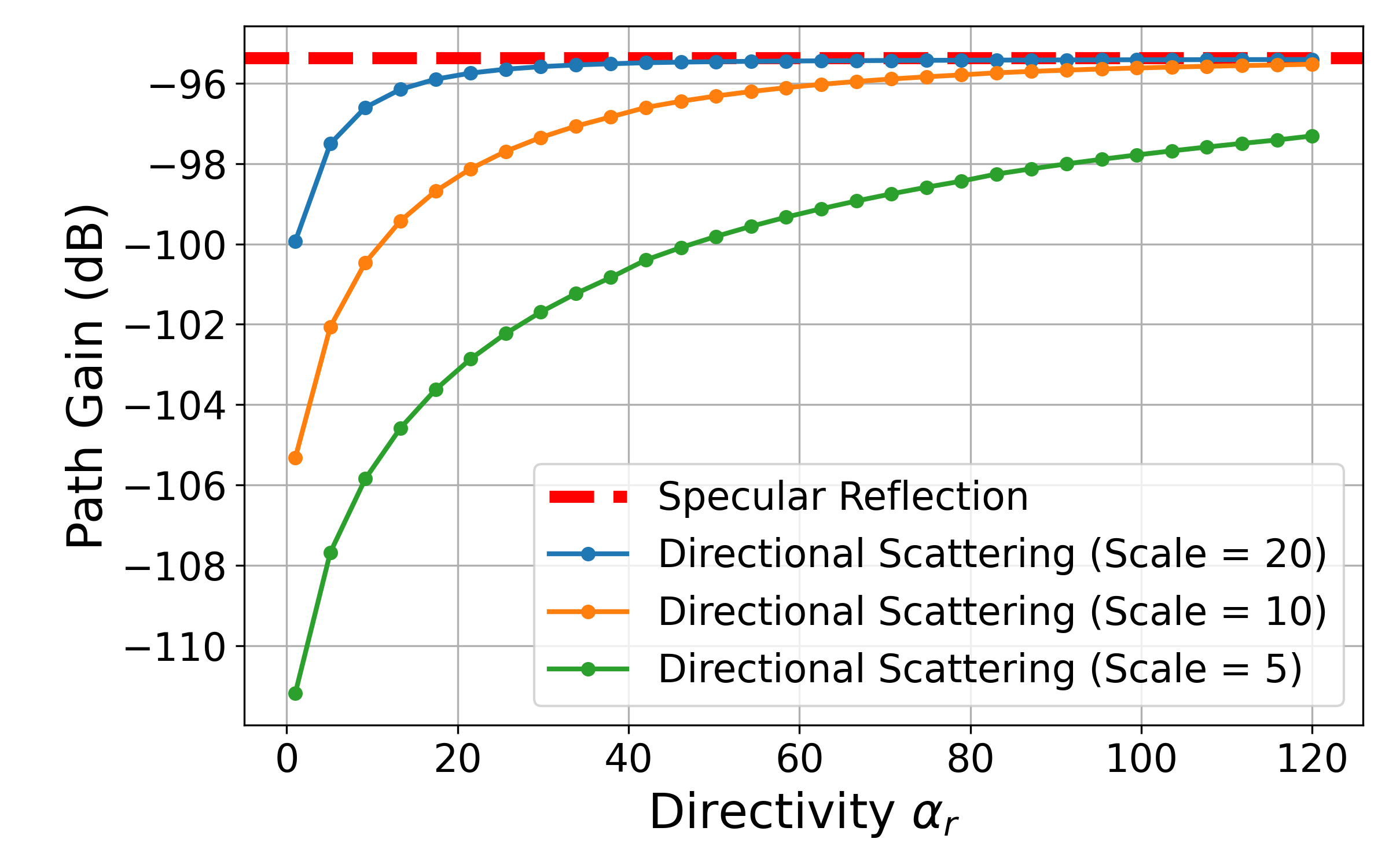}
        \caption{Deviation analysis for different directivities and scales of the reflective surface.}
        \label{fig:error_directivity}
    \end{subfigure}
    \hfill
    \begin{subfigure}[t]{0.32\textwidth}
        \centering
        \includegraphics[width=\textwidth]{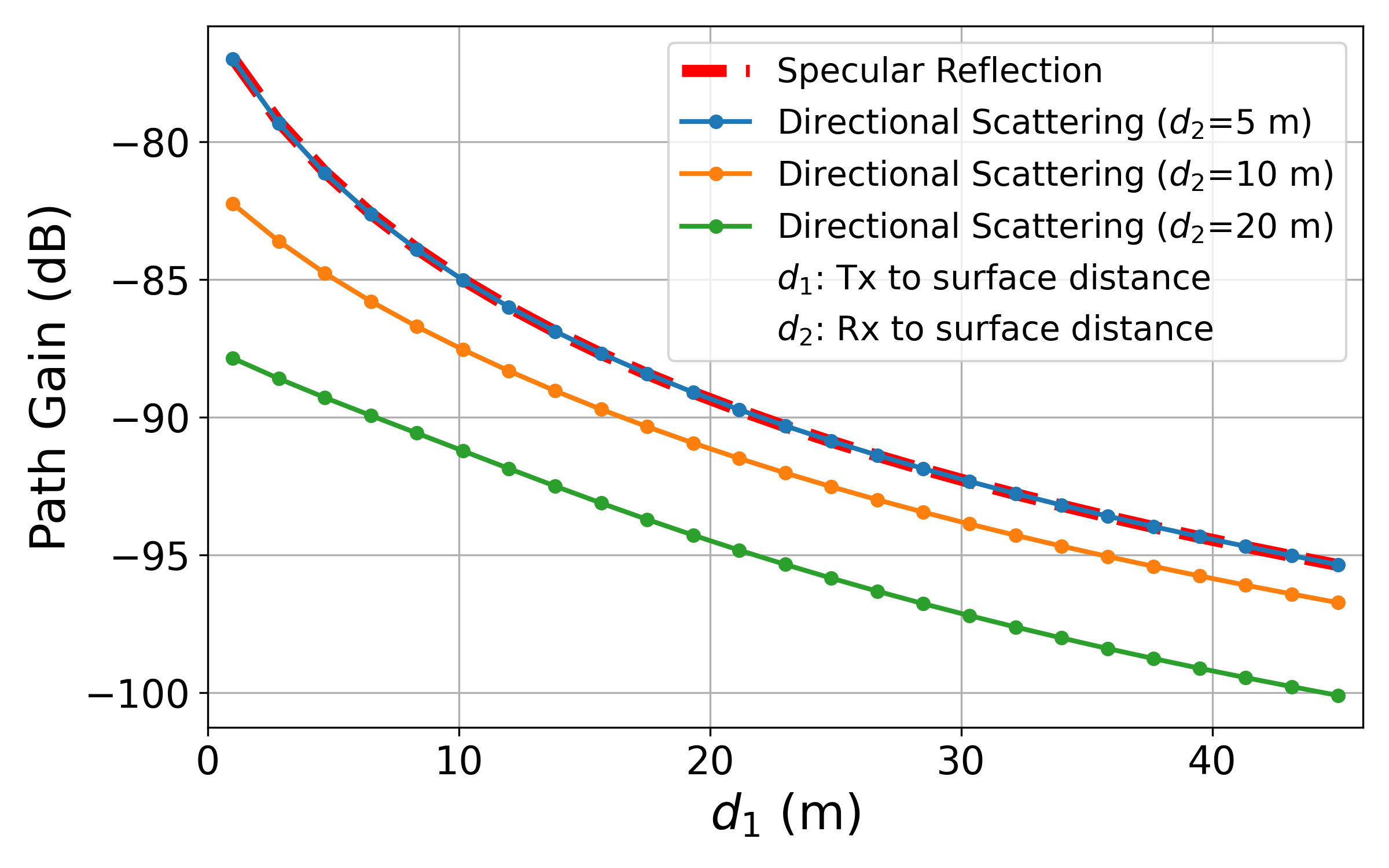}
        \caption{Deviation analysis for distances between TX/RX and reflection point.}
        \label{fig:error_distance}
    \end{subfigure}
    \hfill
    \begin{subfigure}[t]{0.32\textwidth}
        \centering
        \includegraphics[width=\textwidth]{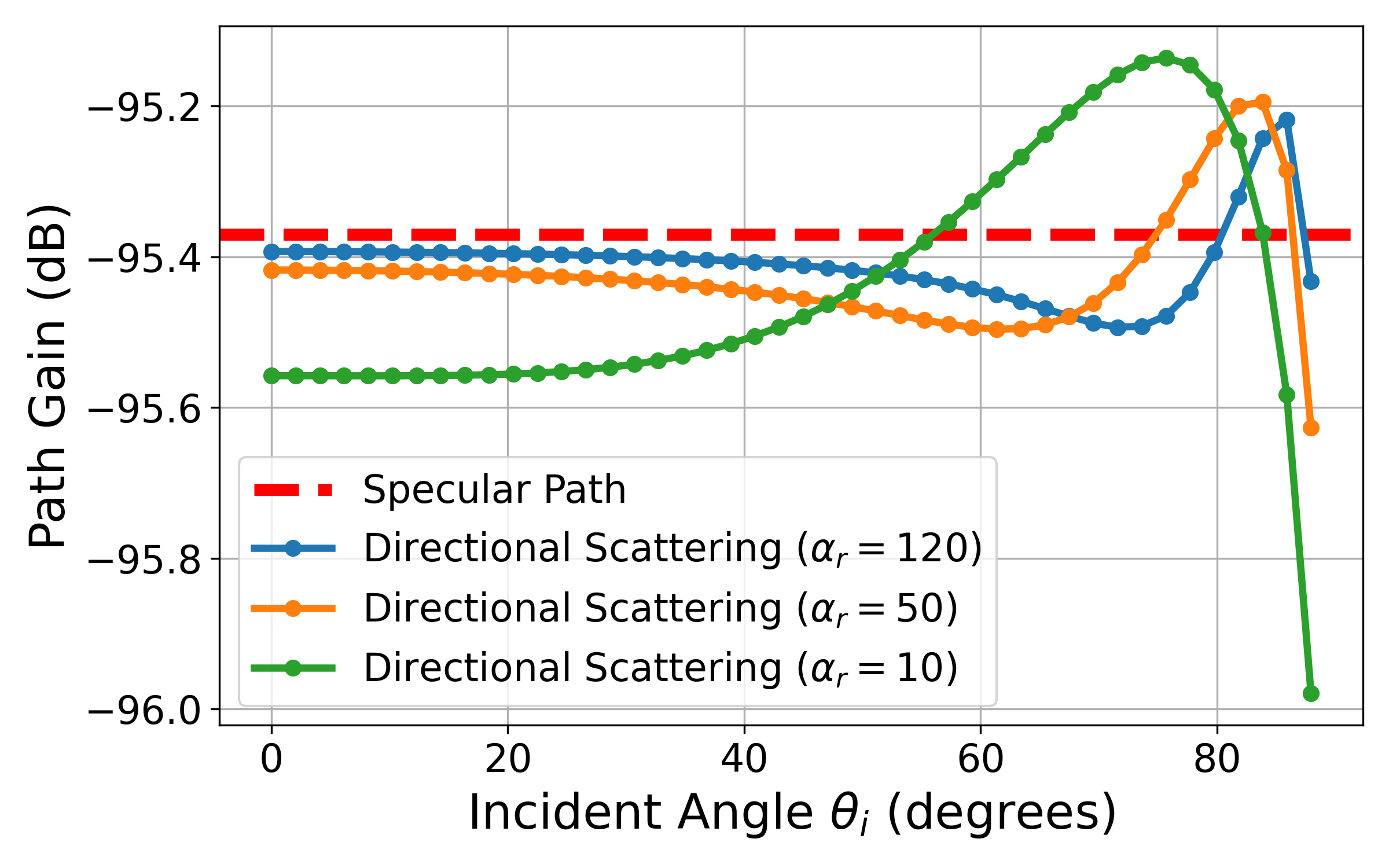}
        \caption{Deviation analysis for different incident angles.}
        \label{fig:error_incident_angle}
    \end{subfigure}
    \caption{Factors that impact the deviation of the directional scattering approximation from the specular reflection model.}
    \label{fig:comparison_error}
\end{figure*}

\begin{figure*}[t]
    \centering

    \begin{minipage}[t]{0.32\textwidth}
        \centering
        \includegraphics[width=\linewidth]{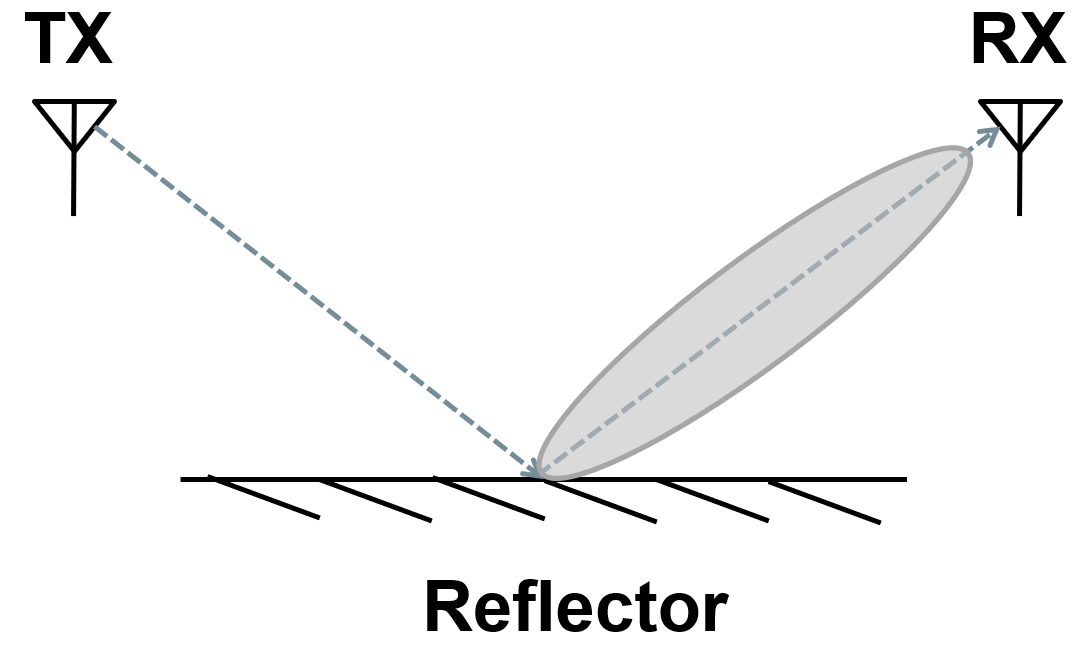}
        \vspace{12pt}
        \caption{Single planar reflector scene for deviation analysis.}
        \label{fig:single_reflector_scene_suppl}
    \end{minipage}
    \hfill
    \begin{minipage}[t]{0.63\textwidth}
        \centering
        \begin{subfigure}[t]{0.48\textwidth}
            \centering
            \includegraphics[width=\linewidth]{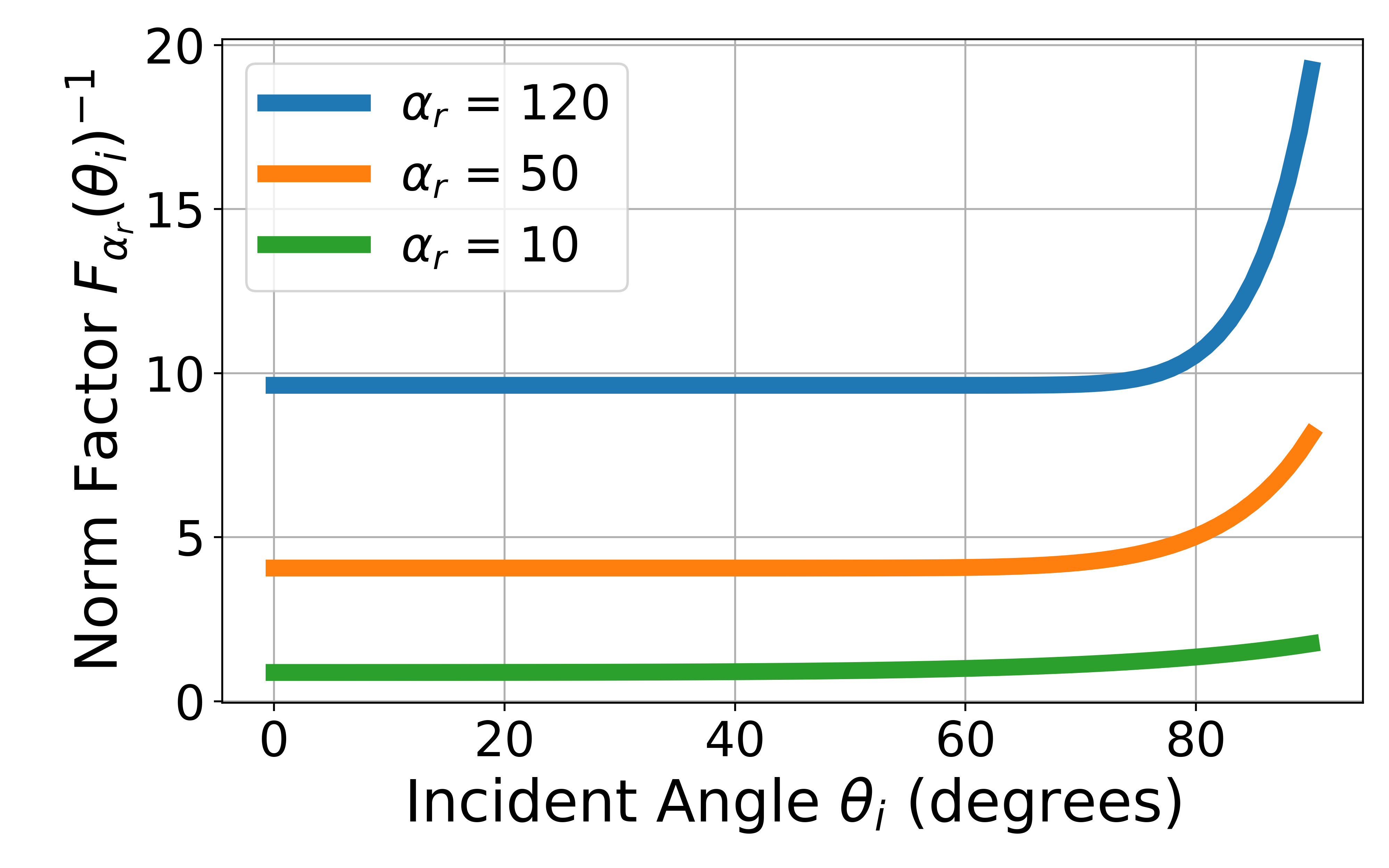}
            \caption{Normalization factor increases with large incident angles.}
            \label{fig:normalization_factor}
        \end{subfigure}
        \hfill
        \begin{subfigure}[t]{0.48\textwidth}
            \centering
            \includegraphics[width=\linewidth]{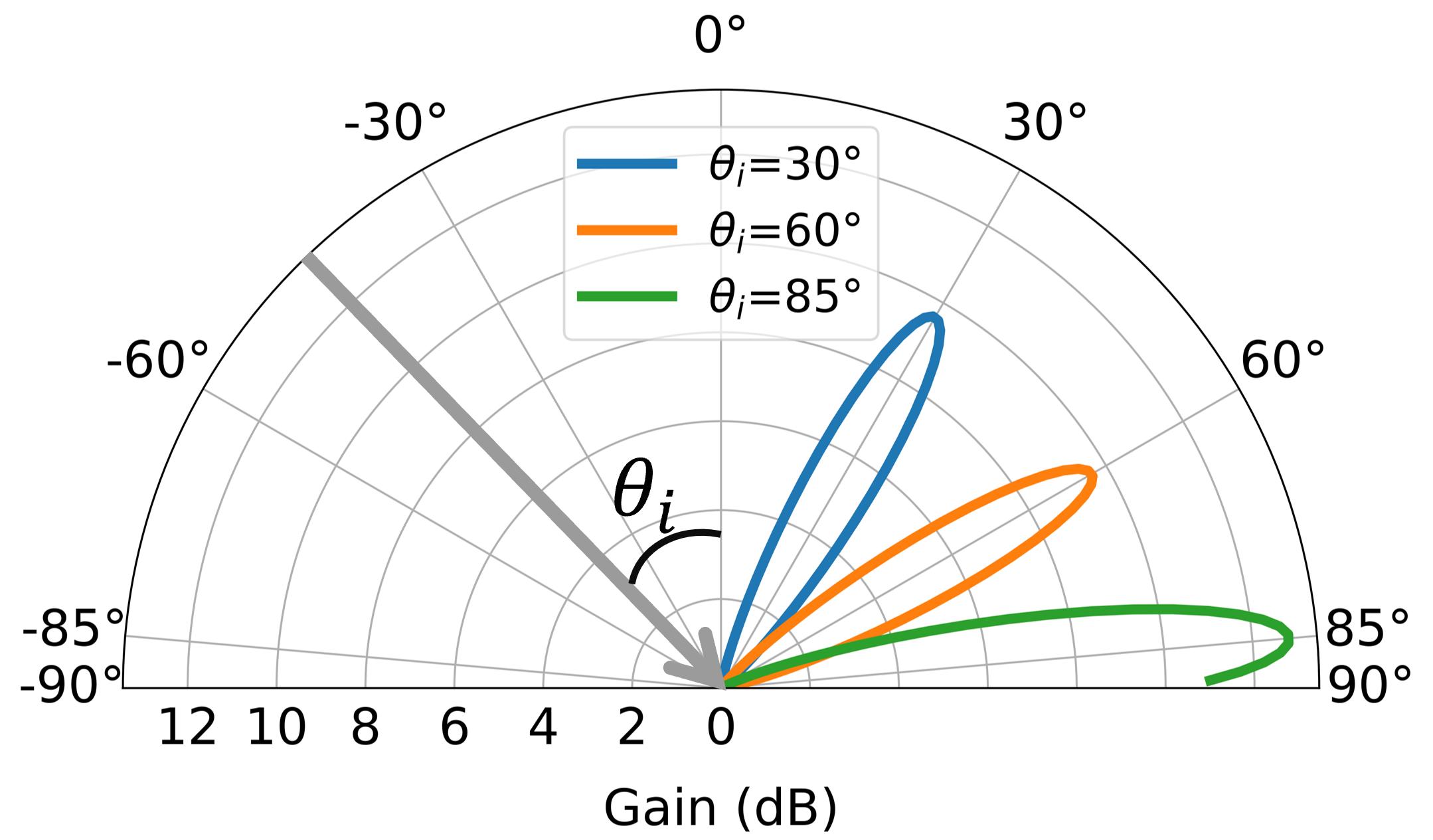}
            \caption{Scattering pattern for different incident angles ($\alpha_r=120$).}
            \label{fig:pattern_incident_angle}
        \end{subfigure}

        \caption{Large incident angles cause a higher normalization factor, which amplifies the scattering lobe.}
        \label{fig:normalization_and_pattern}
    \end{minipage}

\end{figure*}
In the previous section, we have proved that under extreme conditions (i.e., sufficiently large surface and narrow beam), the directional scattering approximation converges to the specular reflection model. In this section, we analyze the deviation introduced by the scattering approximation under non-extreme scenarios and elucidate the source of deviations. Note that we call this "deviation" instead of "error" because the specular reflection model itself is an approximation of the real-world reflection behavior, and is derived under the infinitely large, perfectly flat surface assumption. Here, we treat the specular reflection model as a reference to understand the behavior of \name's approximation under different conditions.

We found that the deviation is mainly caused by four factors: the scale of the reflective surface, the directivity of the scattering lobe, the distance between the surface and the TX/RX, and the incident angle. \cref{fig:error_directivity} shows the deviation versus the directivity of the scattering lobe for different surface scales. We conduct this experiment in the setup shown in ~\cref{fig:single_reflector_scene_suppl}. We fixed the distance between the TX/RX and the reflection point to 25~m and the incident angle to $45^\circ$. By sweeping across different directivities, we observe that the deviation decreases as the directivity increases. This is expected since a higher directivity means a narrower lobe, which better approximates the specular reflection. We also observe that for a given directivity, a larger surface scale results in a smaller deviation. This is because a larger surface can capture more energy, leading to a higher path gain that is closer to the specular reflection. For the scale of 5, the deviation is 2~dB even when the directivity has reached 120, indicating that the surface is not large enough to fully capture the energy. 

\cref{fig:error_distance} shows the deviation versus the distance between the TX/RX and the reflection point. In this experiment, we fixed the surface scale to 5 and directivity to 120, and changed the Tx-to-surface distance, $d_1$, while keeping the Rx-to-surface distance, $d_2$, fixed. The incident angle is also fixed to $45^\circ$. As shown, the deviation increases as the distance increases. This is because, for a fixed surface scale, as the TX/RX moves further away, the effective area of the surface that is illuminated by the TX and visible to the RX decreases, leading to a lower path gain.

The above two experiments reveal that the key factor that affects the deviation is the effective area of the surface that is illuminated by the TX and visible to the RX. When the effective area is large enough to capture all the energy, the directional scattering approximation matches the specular reflection path gains very well. Hence, the directional scattering model can be regarded as a softening of the specular reflection model for finite, non-ideal surface scenarios.

Interestingly, we found that the incident angle also affects the deviation. In this experiment, we fixed the surface scale to 200 (sufficiently large), and the distance between the TX/RX and the reflection point to 25m, then swept the incident angle $\theta_i$ from 0 to $90^\circ$. We move both TX and RX to keep the incident angle and reflection angle the same, such that we can use the specular reflection path gain as the reference. We choose three directivities: 10, 50, and 120, and plot the deviation versus the incident angle in \cref{fig:error_incident_angle}. We observe that for small incident angles (up to $40^\circ$ for $\alpha_r=10$, and $70^\circ$ for $\alpha_r=50 \textrm{~and~} 120$), the deviation is relatively constant with respect to the incident angle. However, as the incident angle increases over $40\circ$, there is an interesting variation pattern of the directional path gain: the deviation first decreases, then increases again after $70^\circ$, and finally decreases again after $80^\circ$. 


\begin{figure}[t]
    \centering
    \begin{minipage}[t]{0.6\textwidth}
        \centering
        \includegraphics[width=\textwidth]{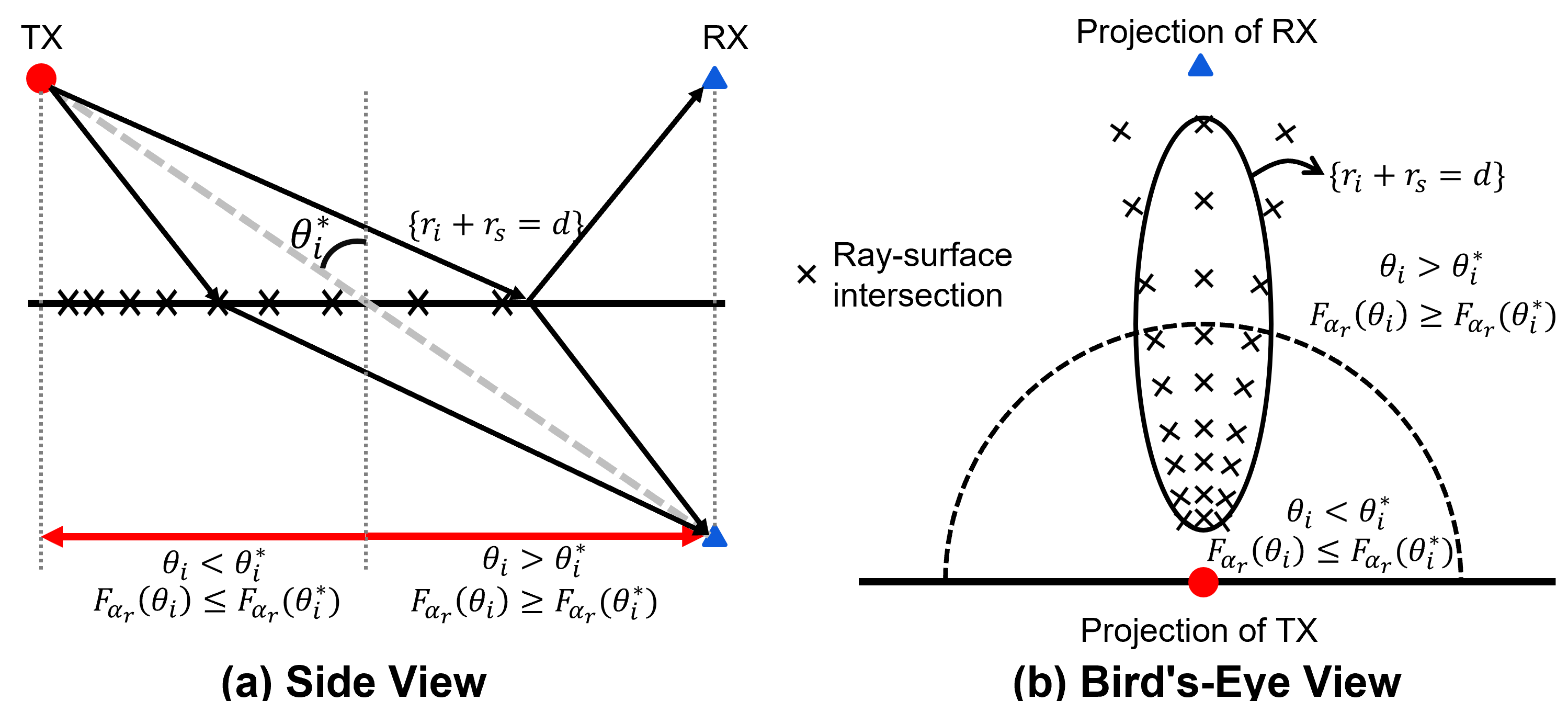}
        \caption{Illustration of the sampling issue that causes the variation with respect to the incident angle (bird's-eye view).}
        \label{fig:incident_angle_deviation_explained}
    \end{minipage}
    \hspace{10pt}
    \begin{minipage}[t]{0.3\textwidth}
        \centering
        \includegraphics[width=\textwidth]{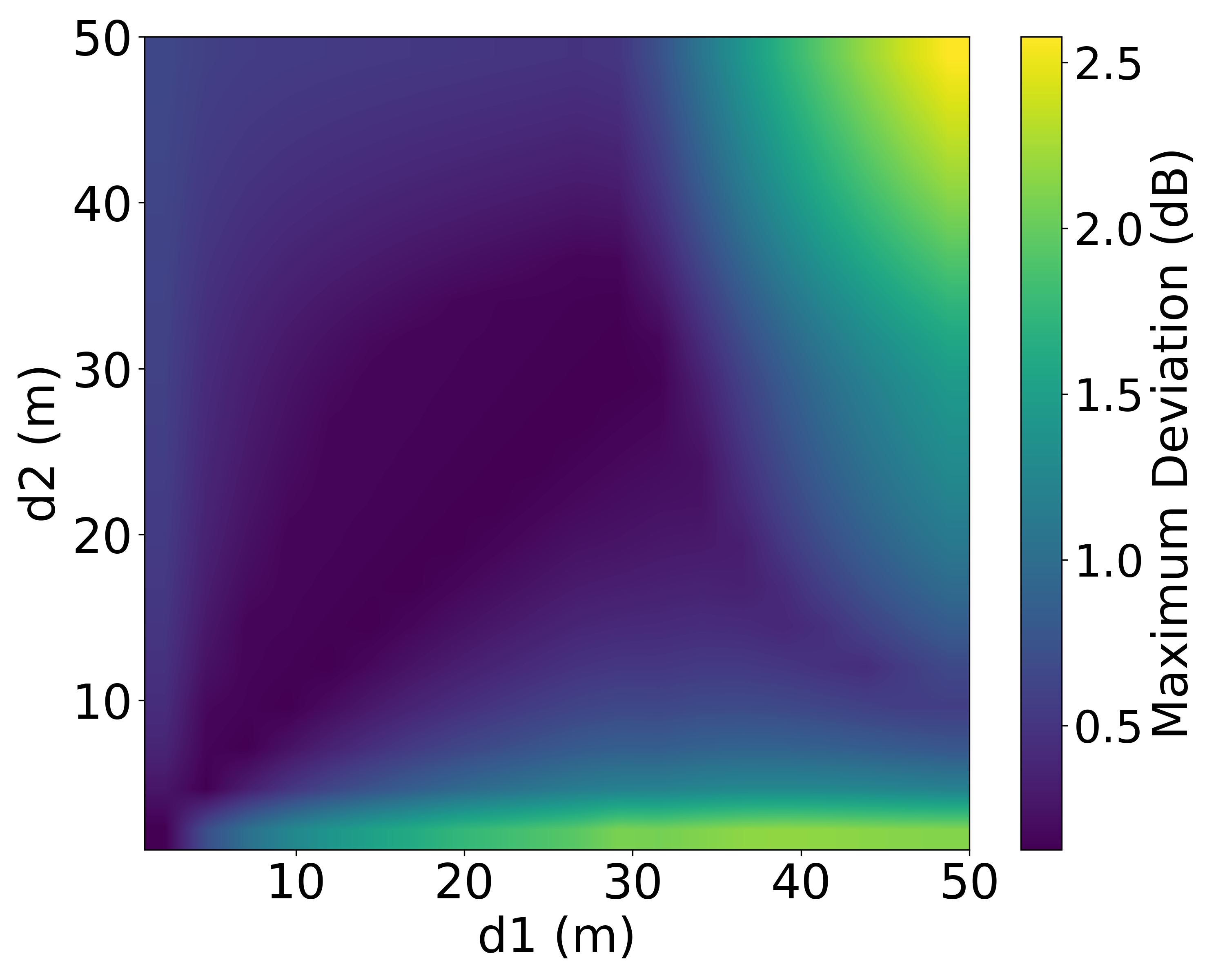}
        \caption{Maximum deviation for different distances between TX/RX and reflection point.}
        \label{fig:path_gain_max_error_heatmap}
    \end{minipage}
\end{figure}


Next, we perform a qualitative analysis to elucidate the reason behind this variation. First, we found that this variation is related to the normalization factor of the scattering lobe. As we discussed in \cref{app:scattering_convergence}, the scattering pattern must satisfy the energy conservation equation (~\cref{eq:energy_cons_scattering}). In our case, we adopt the scattering pattern model in~\cite{degli2007measurement}, which has a normalization factor, $F_{\alpha_r}(\theta_i)^{-1}$, (~\cref{eq:directional_scattering}) to ensure the energy conservation. We plot the normalization factor versus the incident angle in ~\cref{fig:normalization_factor}. We observe that the normalization factors share the same pivot point at around $40^\circ$ for $\alpha_r=10$ and around $70^\circ$ for $\alpha_r=50 \textrm{~and~} 120$. After the pivot point, the normalization factor increases rapidly as the incident angle increases. This means that for large incident angles, the scattering lobe is amplified more (to compensate for the cutting by the scattering plane) to satisfy the energy conservation equation. \cref{fig:pattern_incident_angle} shows the scattering pattern for different incident angles when $\alpha_r=120$. We can see that as the incident angle increases to $85^\circ$, the scattering lobe is cut more by the scattering plane, and the remaining lobe is amplified more to satisfy the energy conservation equation.

Now that we know that the variation is related to the amplification of the scattering lobe for large incident angles, we further investigate why this amplification causes the variation in deviation. According to \cref{eq:int_rx_power_effective_area}, the path gain of a single scattering path is also related to the scattering path length $(r_i + r_s)$. If we plot the level set of the scattering path length on the scattering surface, the contours will be ellipses, as illustrated in \cref{fig:incident_angle_deviation_explained}. We also plot the curve of $\theta_i=\theta_i^*$ on the figure, where $\theta_i^*$ is the incident angle that results in the specular reflection (our reference). The resulting curve is a circle. Within the circle, $\theta_i < \theta_i^*$, and $F_{\alpha_r}(\theta_i) \leq F_{\alpha_r}(\theta_i^*)$. On the contrary, $\theta_i > \theta_i^*$, and $F_{\alpha_r}(\theta_i) \geq F_{\alpha_r}(\theta_i^*)$ outside the circle. The intersection of the circle and the ellipses divides the graph into two parts. The part inside the circle has a smaller area, while the part outside the circle has a larger area. To this end, the sampling density comes into the scope, since the number of samples in each area and their strengths will both affect the final integration result of the path gain. In our implementation, we sample the TX radiation sphere uniformly, which results in a non-uniform sampling on the scattering surface. The closer to the TX's projection, the denser the sampling. Therefore, the part inside the circle has more samples compared to the part outside the circle. This results in two contradictory forces affecting the path gain integration: the part inside the circle has a smaller area, less amplification, but more samples, while the part outside the circle has a larger area, more amplification, but fewer samples. Jointly, these two forces lead to the interesting variation pattern of the deviation with respect to the incident angle.

In summary, the variation is an artifact of the non-uniform sampling on the scattering surface. We identify this as a future research direction for better sampling strategies to fix this artifact. Finally, we investigate the maximum deviation caused by this artifact. In this experiment, we sweep the TX/RX to reflection point distance from 1~m to 50~m, and for each combination, we sweep the incident angle from 0 to $90^\circ$ to find the maximum deviation. We plot the heatmap of maximum deviation versus the distances in ~\cref{fig:path_gain_max_error_heatmap}. We note that for most areas of the plot, the maximum deviation is less than 1~dB, which is good enough for coverage estimation. The maximum deviation increases for larger distances, but it is still bounded by 3~dB even when both distances are 50~m. This indicates that the artifact is not significant in practical scenarios. 

\newpage
\section{Training and Model Details}
\label{app:train_details}
Here, we provide additional details regarding the training procedure and model architecture used in our experiments. We train our \name model using the Adam optimizer with an initial learning rate of $5\times 10^{-4}$. In each scene, we train the model for 100 epochs, which takes approximately 4 hours on a single NVIDIA RTX 3090 GPU. We use a single receiver per batch due to the current architecture limitation. We expect that using multiple receivers per batch can further stabilize and accelerate the training process. 

Our model architecture consists of a 4-layer MLP with 64 hidden units per layer. We use a 6-octave positional encoding to encode the 3D coordinates of the ray-surface interaction points before feeding them into the MLP. The output of the MLP is clipped and transformed using a experiential function to the expected range of material parameters, i.e., the relative permittivity $\epsilon_r \in [1, 200]$ and the conductivity $\sigma \in [10^{-3}, 10^6]$ S/m; the scattering coefficient and cross-polarization discrimination parameters are both in the range of $[0, 1]$. In total, our model has 29,700 learnable parameters. We found that this architecture strikes a good balance between model capacity and training efficiency for our task. 

For the Sionna baseline models, we follow the training procedure in ~\cite{diff-rt-calibration} and train the model for 10000 iterations. We observe that the learning curves have converged at this point. For the NeRF\textsuperscript{2} baseline, we follow the training procedure in ~\cite{NeRF2} and train the model for 30,000 iterations to converge. 

\newpage
\section{Dataset Details}
\label{app:dataset_details}
\begin{figure}[h]
    \centering
    \includegraphics[width=0.7\textwidth]{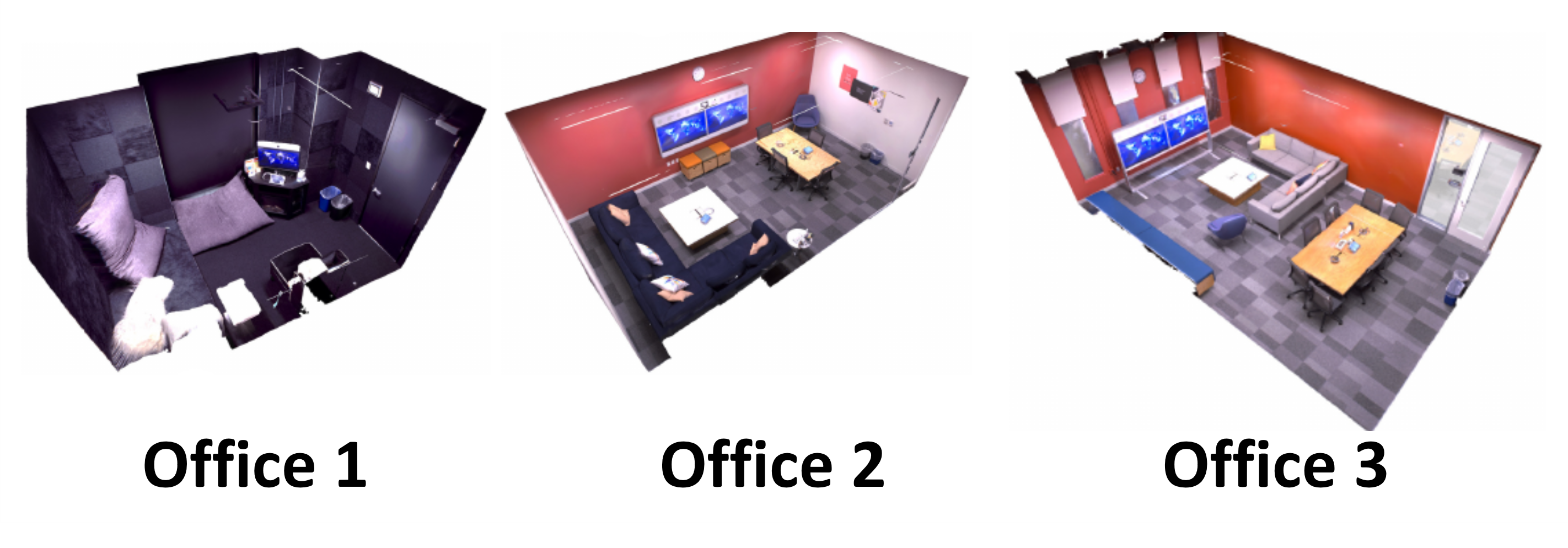}
    \caption{Additional Replica scenes used for evaluation: Office 1, Office 2, and Office 3.}
    \label{fig:add_replica_scenes}
\end{figure}
\begin{figure}[h]
    \centering
    \includegraphics[width=\linewidth]{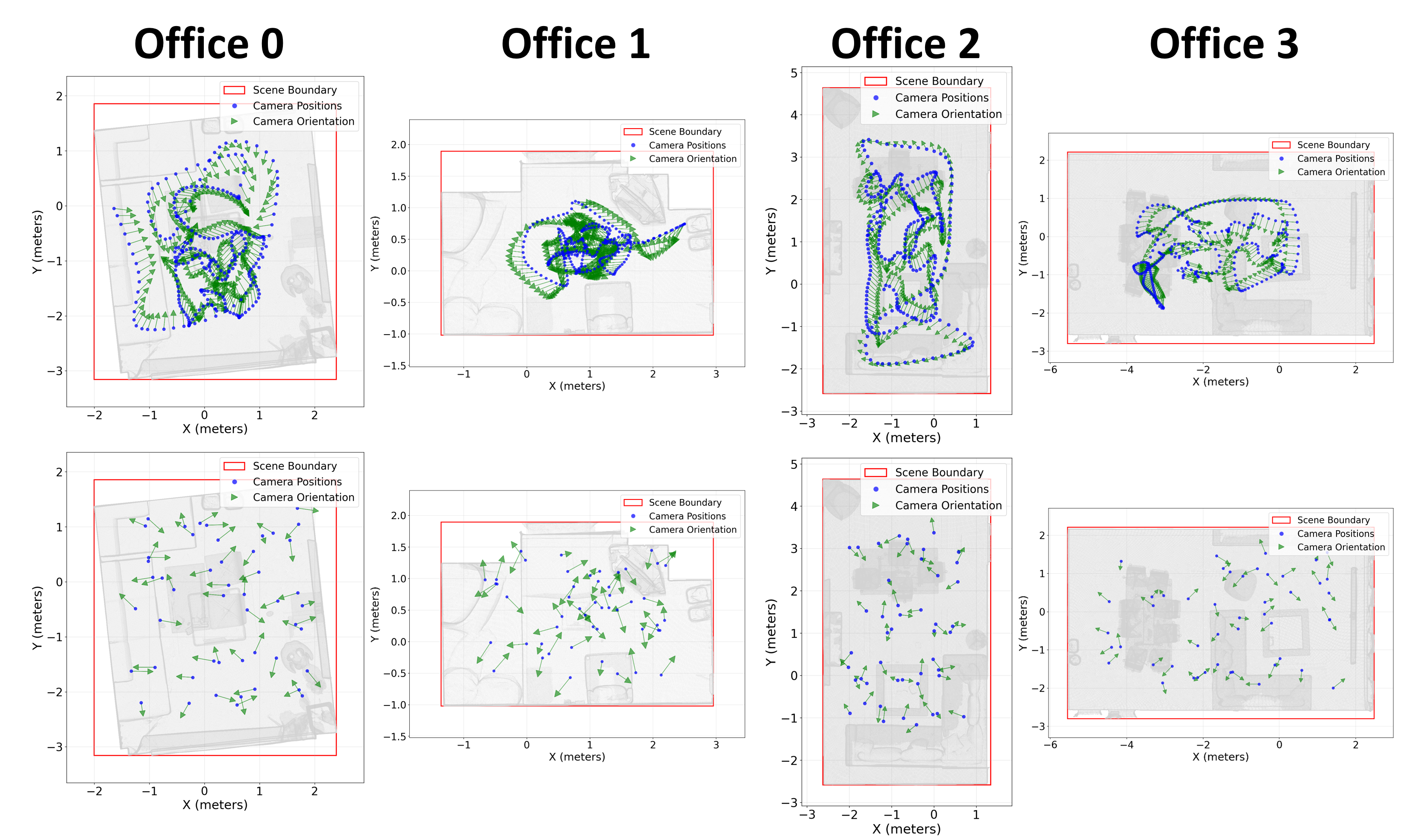}
    \vspace{6pt}
    \includegraphics[width=\linewidth]{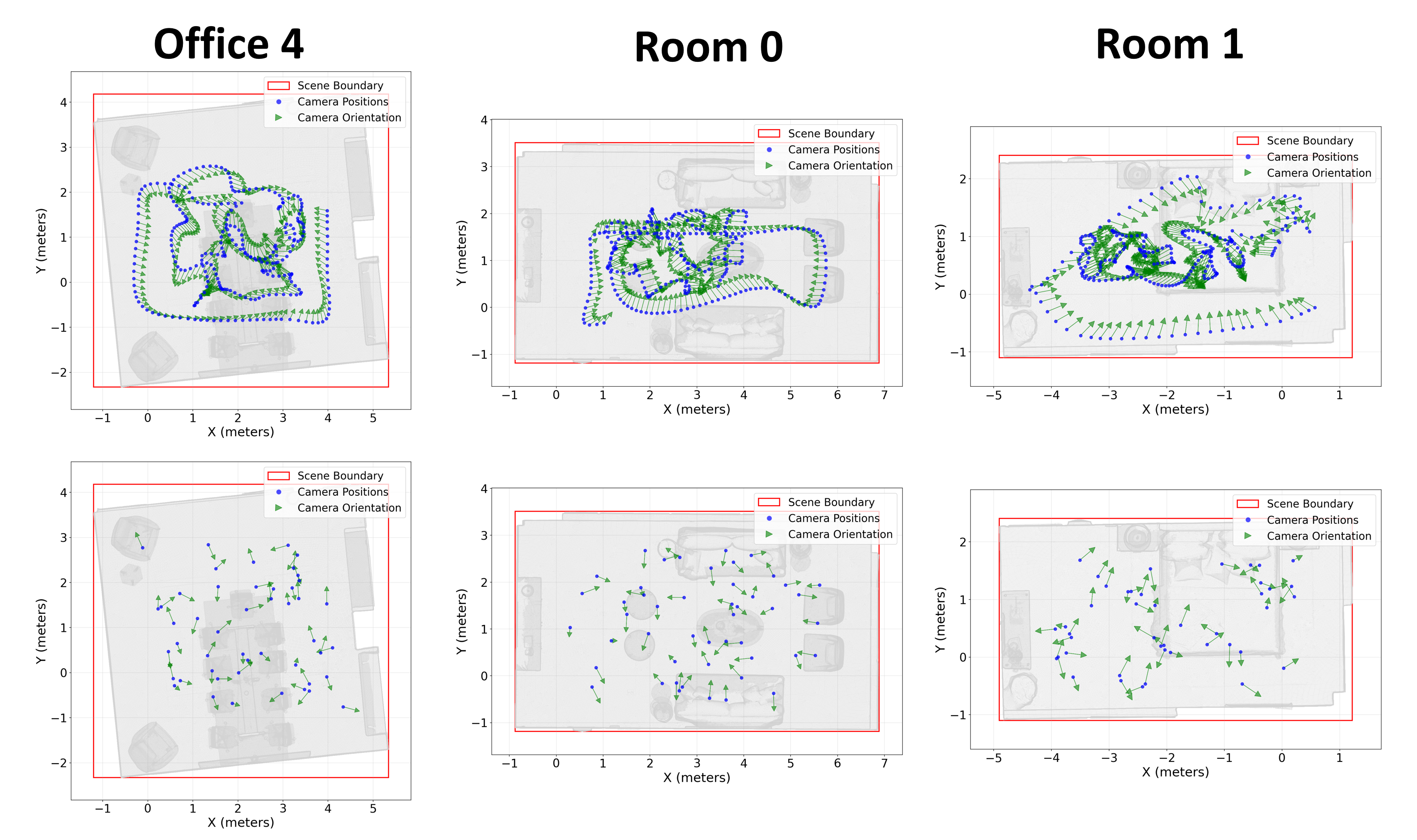}
    \caption{Sample poses in the TRAJ and RAND datasets (subsampled by 10x for visualization). Top: TRAJ poses follow a camera trajectory around the scene, used for training. Bottom: RAND poses are randomly distributed around the scene, used for testing. Note this is the 2D projection for visualization, the actual positions and orientations are in 3D.}
    \label{fig:dset_traj}
\end{figure}

Here, we present three additional Replica scenes, making the total dataset cover seven realistic indoor environments. The additional scenes, Office 1, Office 2, and Office 3, are illustrated in ~\cref{fig:add_replica_scenes}. We follow the same dataset generation and experimental procedures as described in the main paper. The camera/receiver poses for the TRAJ and RAND datasets are visualized in ~\cref{fig:dset_traj}. For the sake of clarity, we down-sample the poses by 10x in the visualization. In each scene, the TRAJ dataset contains 2,800 training samples, and the RAND dataset contains around 500 test samples (we manually removed some out-of-scene poses). Detailed scene information is summarized in ~\cref{tab:scene_info}. We note that this dataset is built on top of the open-source Replica dataset~\cite{replica19arxiv}, with additional radio material assignments, trajectory sampling, and wireless channel simulations. We will make this dataset publicly available to facilitate future research.

\begin{table}[t]
\centering
\caption{Scene Statistics}
\setlength{\tabcolsep}{6pt}
\renewcommand{\arraystretch}{1.15}
\begin{tabular}{l c c c c c}
\hline
Scene Name & Dimensions (m) & \# Objects & \# Materials & \# Test Samples & Training Sample Density (\#samples/m$^3$) \\
\hline
Office 0 & 4.4 $\times$ 5.0 $\times$ 3.0   & 68  & 10 & 489 & 42.4\\
Office 1 & 4.3 $\times$ 2.9 $\times$ 2.7  & 52  & 8  & 500 & 81.1\\
Office 2 & 4.0 $\times$ 7.2 $\times$ 2.8  & 94  & 8  & 468 & 35.3\\
Office 3 & 8.0 $\times$ 5.0 $\times$ 3.1 & 113 & 7  & 476  & 22.6\\
Office 4 & 6.5 $\times$ 6.5 $\times$ 2.8  & 71  & 7  & 500 & 23.3\\
Room 0   & 7.8 $\times$ 4.7 $\times$ 2.8  & 94  & 8  & 494 & 27.3\\
Room 1   & 6.1 $\times$ 3.5 $\times$ 2.8  & 57  & 8  & 474 & 47.4\\
\hline
\end{tabular}
\label{tab:scene_info}
\end{table}

\newpage
\section{Geometric Differentiability}
\label{app:geo_diff}

\begin{figure}[t!]
    \centering
    \begin{subfigure}{0.47\columnwidth}
        \centering
        \includegraphics[width=\textwidth]{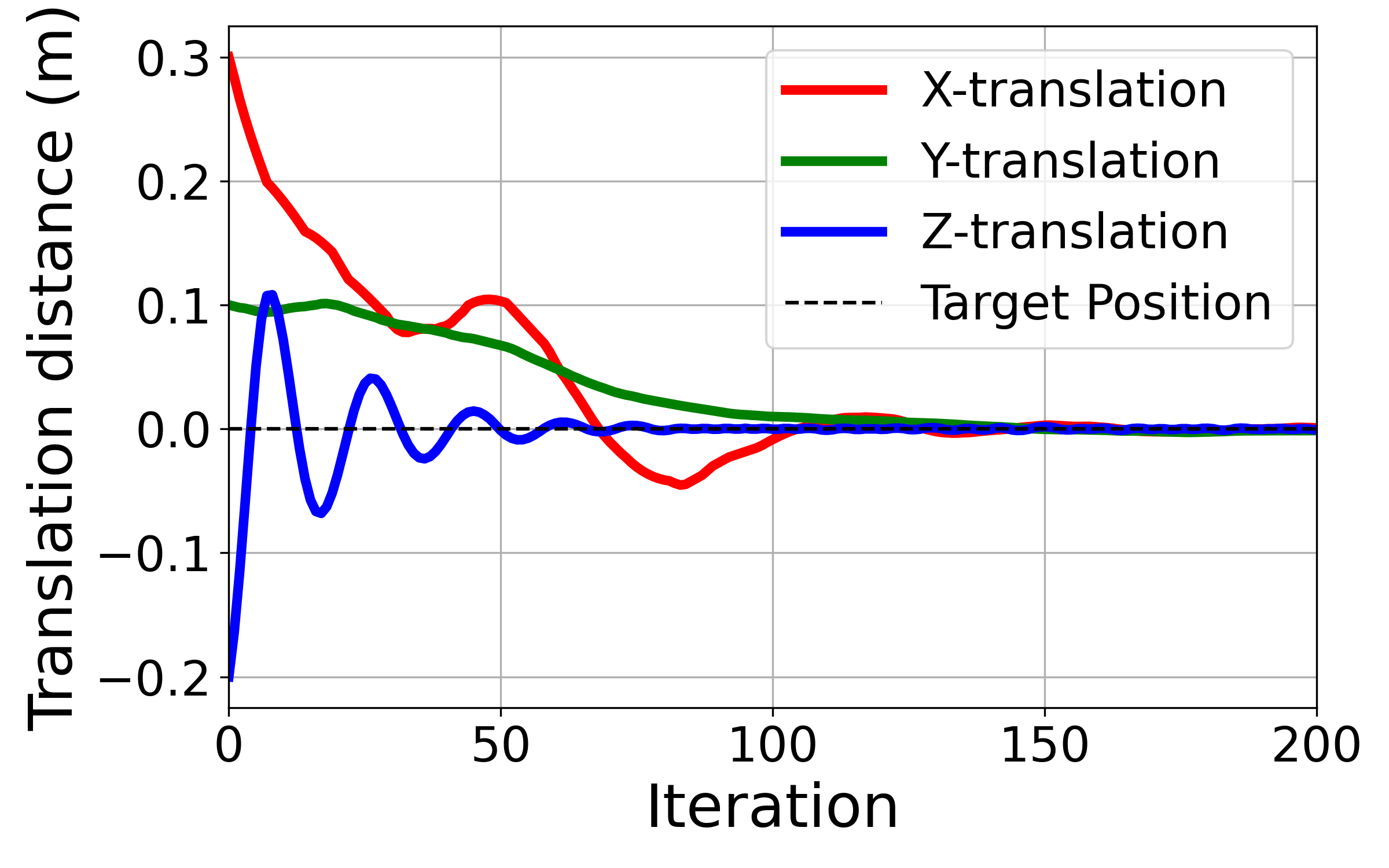}
        \caption{Calibration of translational offsets.}
        \label{fig:calib_trans}
        \vspace{-10pt}
    \end{subfigure}
    \begin{subfigure}{0.47\columnwidth}
        \centering
        \includegraphics[width=\textwidth]{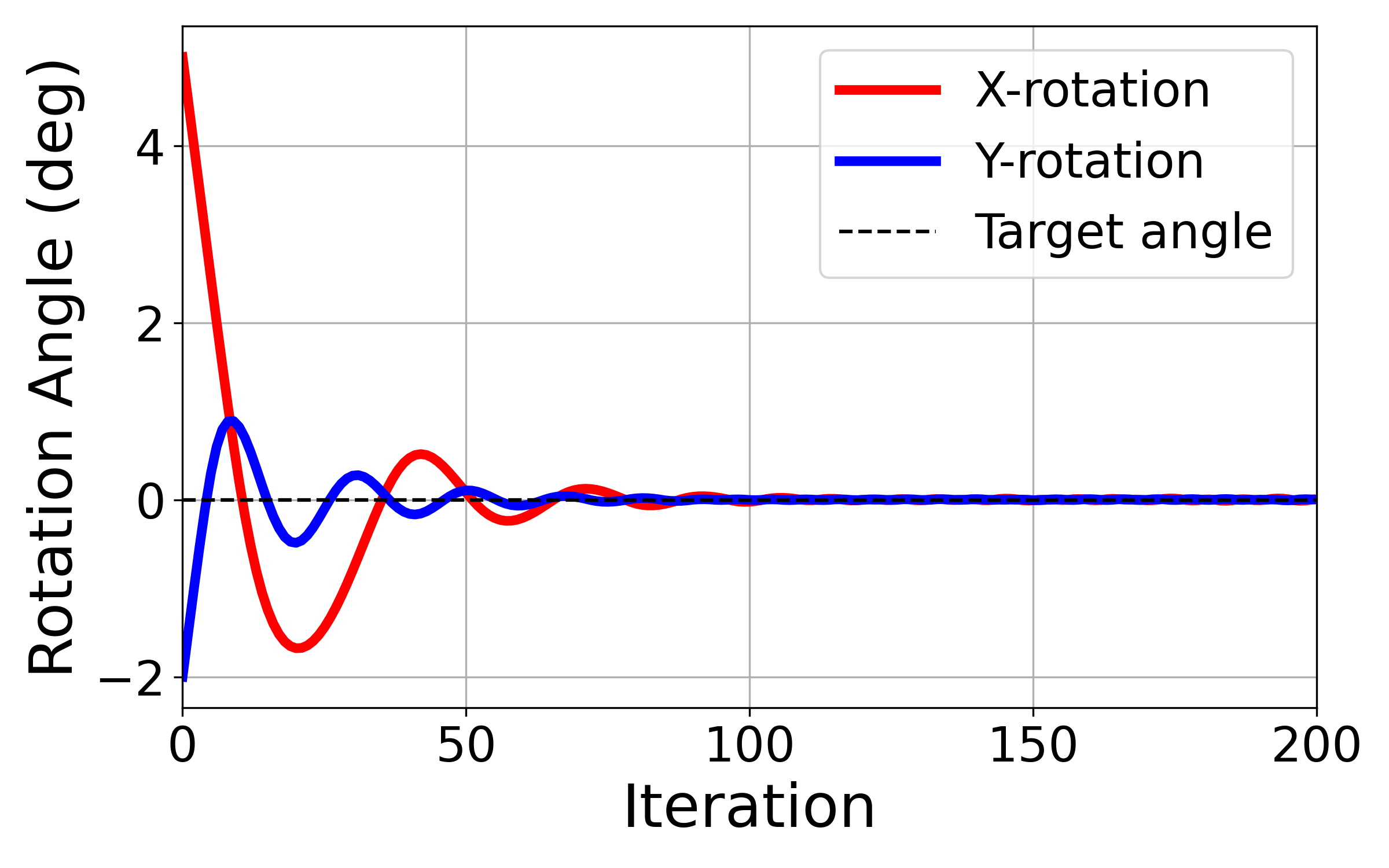}
        \caption{Calibration of rotational offsets.}
        \label{fig:calib_rot}
        \vspace{-10pt}
    \end{subfigure}
    \hfill
    \caption{\name's differentiability enables geometry calibration.}
    \label{fig:calib}
    \vspace{-10pt}
\end{figure}

We demonstrate the differentiability of \name by showing its capability to calibrate simple geometry offsets. We adopt the single-reflector scene shown in \cref{fig:single_reflector_scene}, where both TX and RX arrays are placed symmetrically on the same side of a 1x1~m\textsuperscript{2} planar reflector. We first record the reference AoA power spectrum when the reflector is centered at the origin with $0^\circ$ rotation. We then translate/rotate the reflector by a small distance/angle. And run \name's simulator again to obtain the AoA power spectrum in the distorted scene. Over the course of training, we minimize the log-MAE of AoA power spectrum using Adam optimizer with a learning rate of 0.02. For each experiment, we make the changed geometry variable (translation distance/rotation angle) the sole optimization parameter.  

\cref{fig:calib} shows the calibration of geometry offset over the course of training for 200 iterations. \cref{fig:calib_trans} shows the calibration of offsets along the X/Y/Z axis. As expected, the translation offsets gradually converge to 0 m, correctly recovering the original reflector position. Similarly, we test the calibration of rotational offset along the X/Y axis. As shown in \cref{fig:calib_rot}, in both cases, the calibrated angles converge back to the ground-truth -- 0$^\circ$. We note that rotation around the Z axis (surface normal) has a negligible effect on the AoA power spectrum, leading to vanishing gradients; thus, calibration along this axis is ineffective. 

This experiment shows that \name is differentiable with respect to scene geometry and can successfully correct geometric offsets in a simple scene. However, calibrating large-scale complex scene meshes remains highly challenging, and we deem this as an interesting future direction. 

\newpage
\section{Additional Results}
\label{app:add_results}

Here, we provide the evaluation results on additional scenes to further validate the effectiveness of our proposed \name model. 
Table~\ref{tab:traj_additional} shows the robustness evaluation under geometry noise on three additional scenes. As is shown, our \name model consistently outperforms the Sionna baseline across all metrics, preserving almost perfect prediction accuracy even in the presence of geometry noise. The Average Angular Error (AAE) and Average Power Error (APE) are also consistently lower for \name, demonstrating its robustness to geometry noise.

\begin{table}[h]
\centering
\caption{Robustness under Geometry Noise for Additional Scenes (TRAJ Dataset).}
\setlength{\tabcolsep}{5pt}
\renewcommand{\arraystretch}{1.1}
\begin{tabular}{ll|ccc|ccc}
\hline
Scene & Methods 
& \multicolumn{3}{c|}{Top-1 Peak} 
& \multicolumn{3}{c}{Top-5 Peaks} \\ 
& & F1$\uparrow$ & AAE$\downarrow$ & APE$\downarrow$
  & F1$\uparrow$ & AAE$\downarrow$ & APE$\downarrow$ \\
\hline
\multirow{2}{*}{Office 1}
 & Sionna  & 0.9339 & 2.43 & 0.72 & 0.8576 & 2.23 & 1.14 \\ 
 & \textbf{mmDiff} & \textbf{0.9893} & \textbf{0.36} & \textbf{0.16}
          & \textbf{0.9897} & \textbf{0.33} & \textbf{0.19} \\ 
\hline
\multirow{2}{*}{Office 2}
 & Sionna  & 0.8407 & 2.79 & 0.75 & 0.8120 & 2.48 & 1.30 \\ 
 & \textbf{mmDiff} & \textbf{0.9989} & \textbf{0.18} & \textbf{0.02}
          & \textbf{0.9924} & \textbf{0.19} & \textbf{0.04} \\ 
\hline
\multirow{2}{*}{Office 3}
 & Sionna  & 0.6743 & 4.38 & 1.77 & 0.6570 & 3.77 & 2.15 \\ 
 & \textbf{mmDiff} & \textbf{0.9975} & \textbf{0.25} & \textbf{0.02}
          & \textbf{0.9874} & \textbf{0.29} & \textbf{0.03} \\ 
\hline
\end{tabular}
\label{tab:traj_additional}
\end{table}

\begin{table}[h]
\centering
\caption{Prediction Accuracy Results for Additional Scenes (RAND Dataset).}
\setlength{\tabcolsep}{5pt}
\renewcommand{\arraystretch}{1.1}
\begin{tabular}{ll|ccc|ccc}
\hline
Scene & Methods
& \multicolumn{3}{c|}{Top-1 Peak}
& \multicolumn{3}{c}{Top-5 Peaks} \\
& & F1$\uparrow$ & AAE$\downarrow$ & APE$\downarrow$
  & F1$\uparrow$ & AAE$\downarrow$ & APE$\downarrow$ \\
\hline
\multirow{4}{*}{Office 1}
 & Sionna  & 0.8756 & 2.79 & 0.92 & 0.8222 & 2.57 & 1.57 \\
 & NeRF\textsuperscript{2}  & 0.1433 & 12.45 & 7.95 & 0.4662 & 12.62 & 6.45 \\
 & \name w/o Calib. & 0.7892 & 4.87 & 3.22 & 0.7498 & 3.95 & 3.45 \\
 & \textbf{\name Full} & \textbf{0.9297} & \textbf{1.72} & \textbf{0.80}
          & \textbf{0.9145} & \textbf{1.70} & \textbf{1.02} \\
\hline
\multirow{4}{*}{Office 2}
 & Sionna  & 0.8219 & 2.93 & 0.91 & 0.7845 & 2.57 & 1.74 \\
 & NeRF\textsuperscript{2}  & 0.1496 & 11.36 & 3.80 & 0.2338 & 9.59 & 4.64 \\
 & \name w/o Calib. & 0.7906 & 3.55 & 1.08 & 0.7829 & 3.11 & 1.99 \\
 & \textbf{\name Full} & \textbf{0.9722} & \textbf{1.04} & \textbf{0.27}
          & \textbf{0.9534} & \textbf{0.99} & \textbf{0.47} \\
\hline
\multirow{4}{*}{Office 3}
 & Sionna  & 0.7882 & 4.92 & 4.80 & 0.7376 & 3.61 & 4.83 \\
 & NeRF\textsuperscript{2}  & 0.0273 & 14.09 & 1.38 & 0.0450 & 13.01 & 3.40 \\
 & \name w/o Calib. & 0.8655 & 3.25 & 1.37 & 0.8210 & 2.85 & 1.93 \\
 & \textbf{\name Full} & \textbf{0.9685} & \textbf{1.09} & \textbf{0.32}
          & \textbf{0.9437} & \textbf{0.96} & \textbf{0.50} \\
\hline
\end{tabular}
\label{tab:rand_additional}
\end{table}
Table~\ref{tab:rand_additional} presents the channel prediction accuracy results on additional scenes on the RAND poses dataset. Note that in this experiment, we train all the models on the TRAJ dataset and evaluate them on the RAND dataset to assess their generalization capabilities to unseen poses. We also include an ablation study of our \name model without the calibration, in which we use a randomly initialized neural network for predicting the material parameters. 

The results indicate that our full \name model significantly outperforms both the Sionna and NeRF\textsuperscript{2} baselines and the ablated version without calibration across all metrics. The failure of Sionna is mainly due to its inconsistency in the tracing paths, i.e., the paths traced in the ground truth scene are different from those traced in the reconstructed scene. This inconsistency causes difficulties for the model to learn the material parameters at the ray-surface interactions, leading to poor generalization to unseen poses. NeRF\textsuperscript{2} performs poorly due to the sparsity of training samples. Basically, as presented in~\cite{NeRF2,newrf}, NeRF\textsuperscript{2} requires extremely dense sampling (around 6,000 samples/m$^3$) of the scene to effectively learn the radiation field representation. In our datasets, the sampling density is around 40 samples/m$^3$, which is far below the training condition of NeRF\textsuperscript{2}, leading to its failure in learning a good representation of the scene. The ablated version of \name without calibration also shows inferior performance compared to the full model, highlighting the importance of the calibration step in accurately estimating material parameters for better generalization.

We also provide visual comparisons of the predicted AoA power spectrum at sample poses in ~\cref{fig:viz_comp} to qualitatively demonstrate the superior performance of our \name model in accurately predicting the mmWave channel characteristics compared to the baselines.

\begin{figure}[h]
    \centering
    \includegraphics[width=\textwidth]{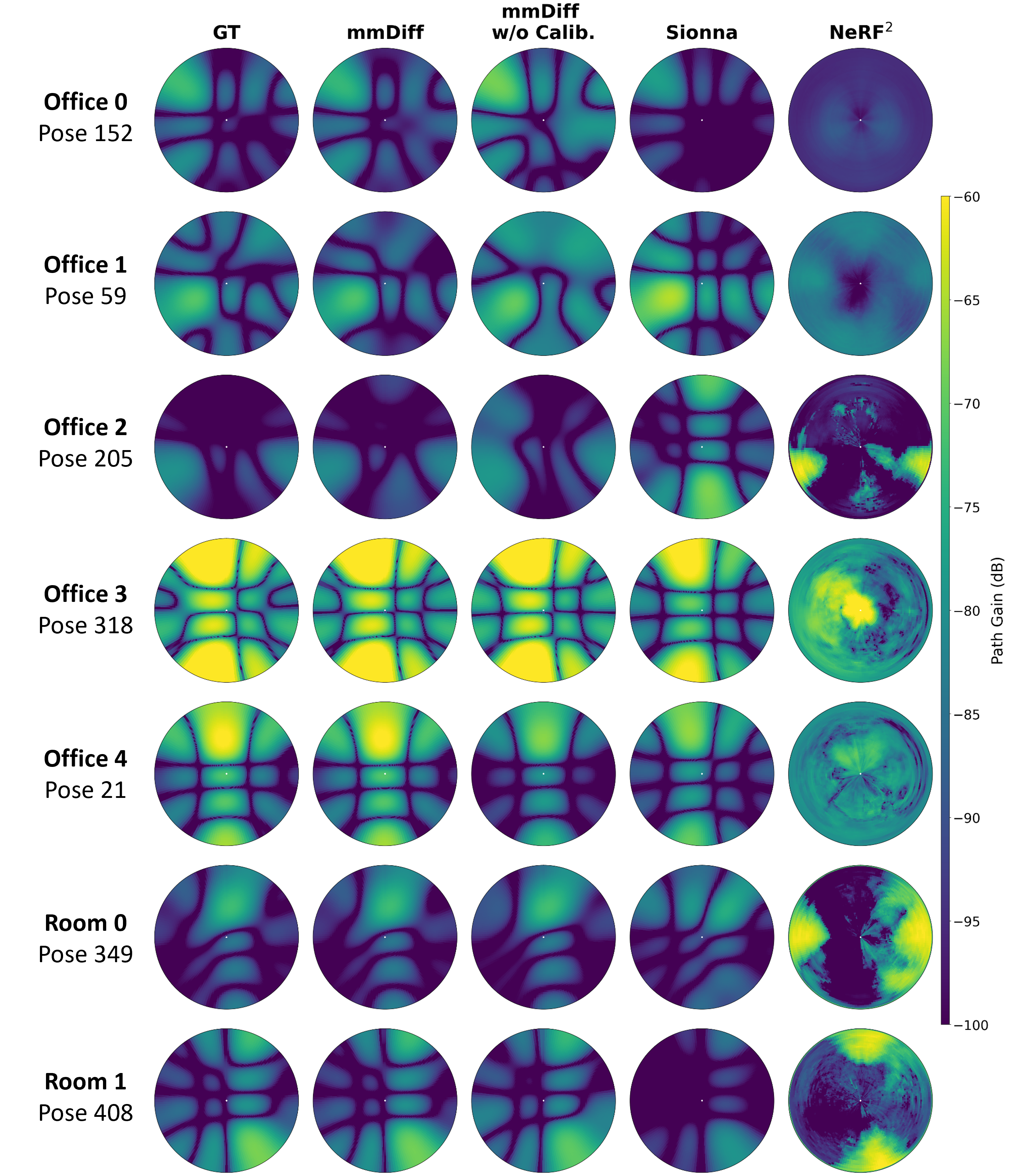}
    \caption{Prediction of AoA power spectrum at sample poses for visual comparison.}
    \label{fig:viz_comp}
\end{figure}

\newpage
\section{Multipath Interference}
\label{app:multipath}
\name adopts a Monte-Carlo-based algorithm to integrate the power of signals arriving from different directions. However, it discards the phase of rays at the ray-surface interaction point and accumulates the real-valued power of the ray weighted by the relative phases across the antenna array.

This is not a problem when there is a single dominant path -- the AoA and power estimation remain correct. This is usually the case in mmWave due to the narrow beam. Even when there is interference from multiple paths, for instance, from a sidelobe, it can be suppressed using the windowing method.

As illustrated in~\cref{fig:multipath_illus}, we consider the single planar reflector case, where both TX and RX are equipped with 8x8 URA, both facing the reflector. There are two paths: the LoS and the reflection path, each with different amplitude and phase. \cref{fig:multipath} shows the AoA spectra produced by the coherent sum of two paths versus the incoherent sum produced by \name. We apply a Hanning window to both TX and RX precoding vectors to reduce the side-lobe level.  Note there is a minor difference between the two spectra, highlighted via yellow boxes. This is due to the inconsistent phase between the two ways of calculating the channel. 

\begin{figure}[h]
    \centering
    \begin{minipage}{0.4\textwidth}
        \centering
        \includegraphics[width=\linewidth]{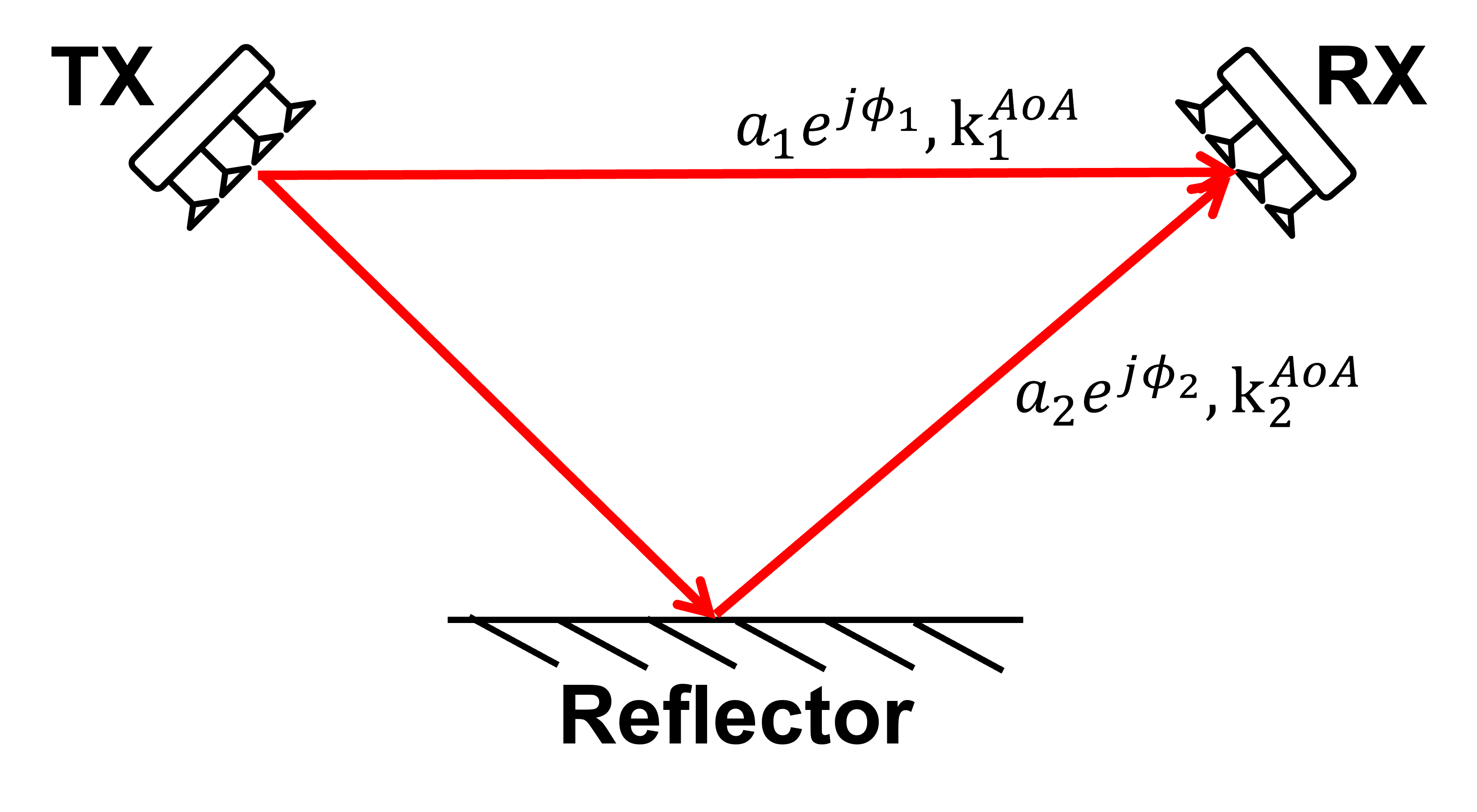}
        \caption{Illustration of the multipath interference}
        \label{fig:multipath_illus}
    \end{minipage}
    \hfill
    \begin{minipage}{0.55\textwidth}
        \centering
        \includegraphics[width=\linewidth]{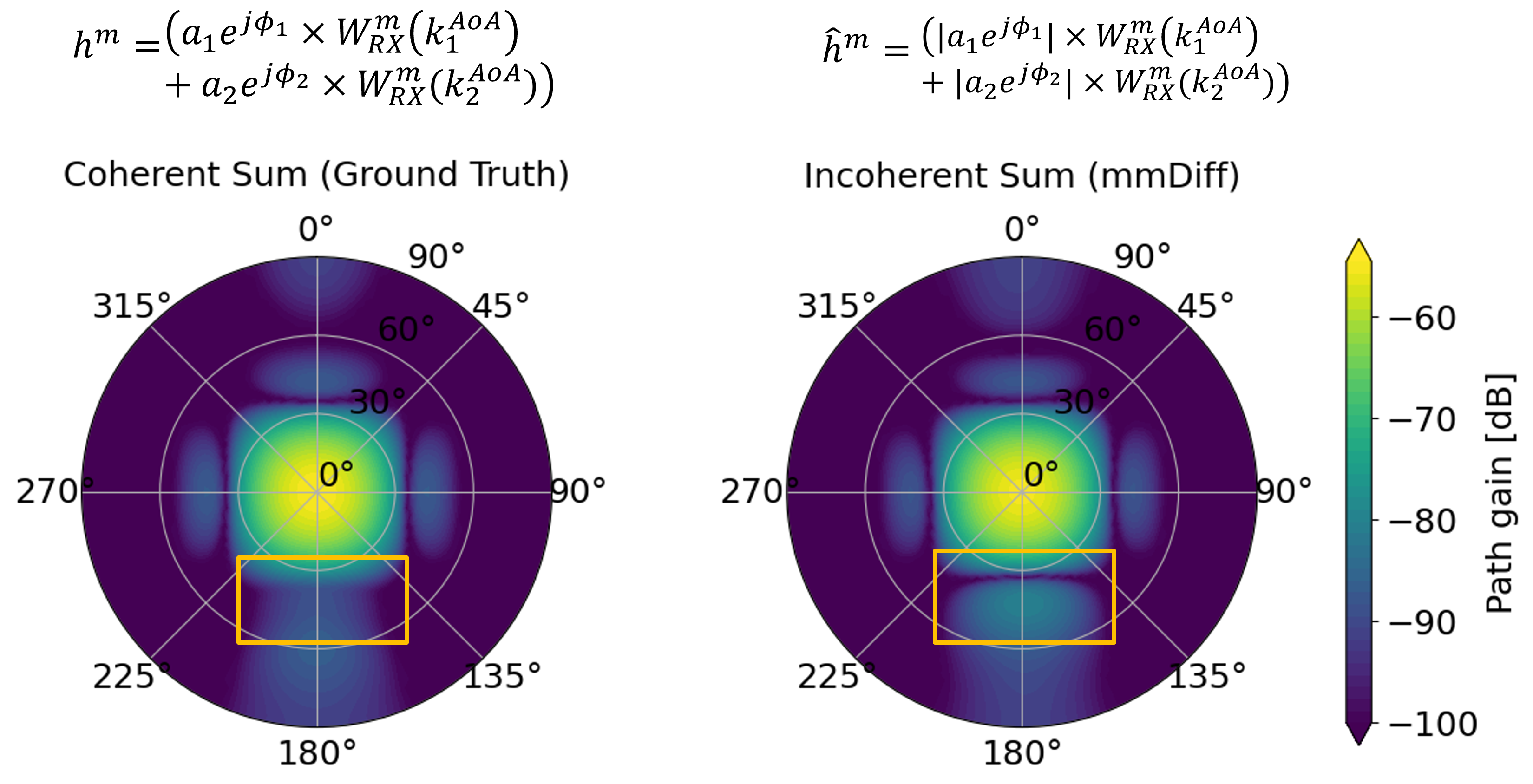}
        \caption{Multipath can cause inconsistent AoA spectra due to the loss of phase information}
        \label{fig:multipath}
    \end{minipage}
\end{figure}

Despite the inconsistency, the calibration works fine. As shown in~\cref{fig:calib_domin_path}, since the reflection path (at 0$^\circ$) is much stronger than the LoS path (at 45$^\circ$), the loss is mostly contributed by the dominant path. By calibrating against the ground truth reference map, we can recover the path gain. 
\begin{figure}[h]
    \centering
    \includegraphics[width=0.8\linewidth]{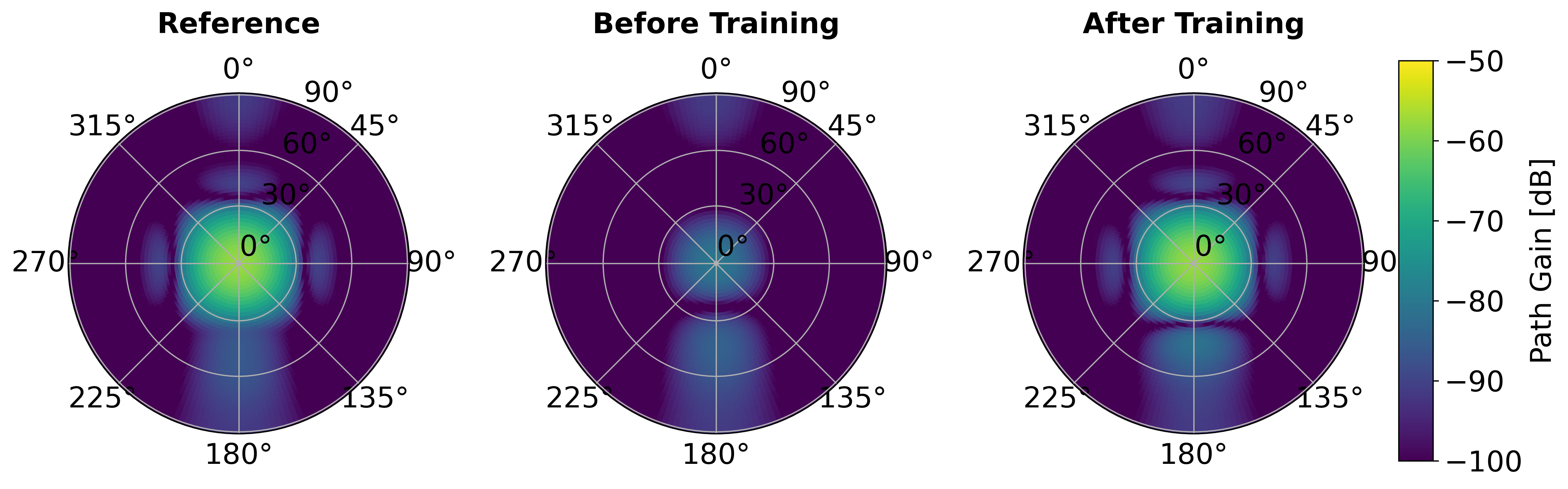}
    \caption{Material calibration with single dominant path.}
    \label{fig:calib_domin_path}
\end{figure}

\end{document}